\documentclass[english]{article}
\usepackage[T1]{fontenc}
\usepackage[latin9]{inputenc}
\usepackage{color}
\usepackage{babel}
\usepackage{amsmath}
\usepackage{amsthm}
\usepackage{amssymb}
\usepackage{graphicx}
\usepackage[unicode=true,pdfusetitle,
 bookmarks=true,bookmarksnumbered=false,bookmarksopen=false,
 breaklinks=false,pdfborder={0 0 1},backref=false,colorlinks=true]
 {hyperref}

\makeatletter
\theoremstyle{plain}
\newtheorem{thm}{\protect\theoremname}
\theoremstyle{definition}
\newtheorem{defn}[thm]{\protect\definitionname}
\theoremstyle{remark}
\newtheorem{rem}[thm]{\protect\remarkname}
\theoremstyle{plain}
\newtheorem{lem}[thm]{\protect\lemmaname}
\theoremstyle{plain}
\newtheorem{cor}[thm]{\protect\corollaryname}
\theoremstyle{plain}
\newtheorem{prop}[thm]{\protect\propositionname}
\theoremstyle{definition}
\newtheorem{example}[thm]{\protect\examplename}

\makeatother

\providecommand{\corollaryname}{Corollary}
\providecommand{\definitionname}{Definition}
\providecommand{\examplename}{Example}
\providecommand{\lemmaname}{Lemma}
\providecommand{\propositionname}{Proposition}
\providecommand{\remarkname}{Remark}
\providecommand{\theoremname}{Theorem}

\begin{document}
\title{Comparison of Markov chains via weak Poincaré inequalities with application
to pseudo-marginal MCMC}
\author{Christophe Andrieu, Anthony Lee, Sam Power, Andi Q. Wang}
\maketitle
\begin{abstract}
We investigate the use of a certain class of functional inequalities
known as weak Poincaré inequalities to bound convergence of Markov
chains to equilibrium. We show that this enables the straightforward
and transparent derivation of subgeometric convergence bounds for
methods such as the Independent Metropolis--Hastings sampler and
pseudo-marginal methods for intractable likelihoods, the latter being
subgeometric in many practical settings. These results rely on novel
quantitative comparison theorems between Markov chains. Associated
proofs are simpler than those relying on drift/minorization conditions
and the tools developed allow us to recover and further extend known
results as particular cases. We are then able to provide new insights
into the practical use of pseudo-marginal algorithms, analyse the
effect of averaging in Approximate Bayesian Computation (ABC) and
the use of products of independent averages, and also to study the
case of lognormal weights relevant to particle marginal Metropolis--Hastings
(PMMH).
\end{abstract}
\global\long\def\dif{\mathrm{d}}%
\global\long\def\Var{\mathrm{Var}}%
\global\long\def\R{\mathbb{R}}%
\global\long\def\X{\mathcal{X}}%
\global\long\def\calE{\mathcal{E}}%
\global\long\def\E{\mathsf{E}}%
\global\long\def\Ebb{\mathbb{E}}%

\global\long\def\ELL{\mathrm{L}^{2}}%
\global\long\def\osc{\mathrm{osc}}%
\global\long\def\Id{\mathrm{Id}}%

\section{Introduction}

\subsection{Motivation}

The theoretical analysis of Markov chain Monte Carlo (MCMC) algorithms
can provide twofold benefits for users. On the one hand, it provides
fundamental reassurance and theoretical guarantees for the correctness
of algorithms, and on the other hand, can also offer guidance on parameter
tuning to maximise efficacy. 

Aside from high-dimensional scaling limit arguments \cite{Roberts1997},
two approaches have proved particularly successful for characterizing
the properties of MCMC algorithms \cite{Bakry08}: Lyapunov drift/minorization
conditions \cite{MeynTweedie,rosenthal-1995,Douc18book}, and functional-analytic
tools on Hilbert spaces, in particular in the reversible setup \cite[Chapter 22]{Kontoyiannis03,Douc18book}.
The former have been the most successful for the study of stability
and convergence rates, despite the inherent difficulty of constructing
an appropriate Lyapunov function. A particular success has been the
development of tools to analyse the scenario where the Markov transition
kernel does not possess a spectral gap, and hence converges at a subgeometric
rate (see \cite{Douc18book} for a book-length treatment). In contrast,
functional-analytic tools have been particularly successful at characterising
the resulting asymptotic variance, but their application to characterising
convergence rates has been limited to the scenario where a spectral
gap exists (see \cite{khare-2011} for example). This is despite the
existence of functional-analytic tools such as weak Poincaré or Nash
inequalities, which have been successfully applied to continuous-time
Markov processes in the absence of a spectral gap \cite{Rockner2001}. 

The aim of this paper is to fill this gap, and show how weak Poincaré
inequalities can be particularly useful for analyzing certain MCMC
algorithms and answering pertinent practical questions. Our main focus
here will be on pseudo-marginal algorithms \cite{Andrieu09}, a particular
type of MCMC method for which pointwise unbiased estimates of the
target density are sufficient for their implementation. We show that
weak Poincaré inequalities allow us to significantly expand and greatly
simplify the results of \cite{Andrieu15}, characterising precisely
the degradation in performance incurred when using noisy estimates
of the target density. This is particularly appealing because pseudo-marginal
Markov kernels often do not possess a spectral gap on general state
spaces, either because the noise is unbounded \cite{Andrieu09,Andrieu15}
or because the noise is not uniformly bounded and ``local proposals''
are used \cite{Lee14}, which is fairly common in practice.

To the best of our knowledge, while Nash inequalities for finite state
space Markov chains have been considered in \cite{Diaconis1996},
weak Poincaré inequalities have not received the same attention in
this context and it is not possible to point to a suitable reference
for background. In Section~\ref{sec:Weak-Poincar=0000E9-inequalities-overview}
we provide a comprehensive overview of the theory tailored to the
Markov chain scenario; some of the results given therein are new to
the best of our knowledge. In Section~\ref{sec:Chaining-Poincar=0000E9-inequalities}
we develop a series of new comparison results between Markov chains
sharing a common invariant distribution. In Section~\ref{sec:Application-to-pseudo-marginal}
we apply our results to pseudo-marginal algorithms, providing a simple
and comprehensive theory of the impact of using noisy densities on
the convergence properties of pseudo-marginal algorithms, which we
leverage to clarify implementational considerations. We consider the
effect of averaging, with applications to Approximate Bayesian Computation
(ABC) and when using products of independent averages, and finally
provide an analysis when the weights are lognormal, relevant to the
Particle Marginal MH (PMMH). The proofs not appearing in the main
text can be found in the appendices.

\subsection{Notation}

We will write $\mathbb{N}=\{1,2,\dots\}$ for the set of natural numbers,
and $\R_{+}=(0,\infty)$ for positive real numbers. 

Throughout we will be working on a general measurable space $(\E,\mathcal{F})$. 
\begin{itemize}
\item For a set $A\in\mathcal{F}$, its complement in $\E$ is denoted by
$A^{\complement}$. We denote the corresponding indicator function
by $\mathbb{I}_{A}:\E\to\{0,1\}$.
\item We assume $(\E,\mathcal{F})$ is equipped with a probability measure
$\mu$, and write $\ELL(\mu)$ for the Hilbert space of (equivalence
classes) of real-valued square-integrable measurable functions with
inner product $\langle f,g\rangle=\int_{\E}f(x)g(x)\,\dif\mu(x)$
and corresponding norm $\|\cdot\|_{2}$. We write $\ELL_{0}(\mu)$
for the set of functions $f\in\ELL(\mu)$ which also satisfy $\mu(f)=0$.
\item We assume that the diagonal is measurable in $\mathsf{E}\times\mathsf{E}$,
i.e. $\{(x,x):x\in\mathsf{E}\}\in\mathcal{F}\otimes\mathcal{F}$.
This assumption holds, e.g., on a Polish space endowed with its Borel
$\sigma$-algebra.
\item More generally, for $p\in[1,\infty)$, we write $\mathrm{L}^{p}(\mu)$
for the Banach space of real-valued measurable functions with finite
$p$-norm, $\|f\|_{p}:=\left(\int_{\E}|f|^{p}\,\dif\mu\right)^{1/p}$,
and $\mathrm{L}_{0}^{p}(\mu)$ for $f\in\mathrm{L}^{p}(\mu)$ with
$\mu(f)=0$.
\item For a measurable function $f:\mathsf{E}\to\R$, let $\|f\|_{\osc}:=\mathrm{ess_{\mu}}\sup f-\mathrm{ess}_{\mu}\inf f$.
\item For two probability measures $\mu$ and $\nu$ on $(\E,\mathcal{F})$
we let $\mu\otimes\nu(A\times B)=\mu(A)\nu(B)$ for $A,B\in\mathcal{F}$.
For a Markov kernel $P(x,\dif y)$ on $\E\times\mathcal{F}$, we write
for $\bar{A}\in\mathcal{F}\times\mathcal{F}$, the product $\sigma$-algebra,
$\mu\otimes P(\bar{A})=\int_{\bar{A}}\mu(\dif x)P(x,\dif y)$. 
\item A point mass distribution at $x$ will be denoted by $\delta_{x}(\dif y)$.
\item $\Id:\ELL(\mu)\to\ELL(\mu)$ denotes the identity mapping, $f\mapsto f$.
\item Given a bounded linear operator $T:\ELL(\mu)\to\ELL(\mu)$, we let
$\calE(T,f)$ be the Dirichlet form defined by $\langle(\Id-T)f,f\rangle$
for any $f\in\ELL(\mu)$.
\item For such an operator $T$, we write $T^{*}$ for its adjoint operator
$T^{*}:\ELL(\mu)\to\ELL(\mu)$, which satisfies $\langle f,Tg\rangle=\langle T^{*}f,g\rangle$
for any $f,g\in\ELL(\mu)$.
\item For such an operator $T$, we denote its spectrum by $\sigma(T)$.
\item We will write $a\wedge b$ to mean the (pointwise) minimum of real-valued
functions $a,b$ and $a\vee b$ for the corresponding maximum. For
$s\in\R$, we will write $\left(s\right)_{+}:=s\vee0$ for the positive
part.
\item $\inf A$ denotes the infimum of set $A$ and $\inf\emptyset=\infty$.
\item For a differentiable function $f$, we denote its derivative by ${\rm D}f$.
\end{itemize}

\section{Weak Poincaré inequalities\label{sec:Weak-Poincar=0000E9-inequalities-overview}}

\subsection{General case\label{subsec:General-case}}

\subsubsection{Definitions and basic results}

Throughout this work,\textcolor{black}{{} in analogue with the existing
notions for continuous-time Markov processes \cite{Rockner2001}},
we will call a \emph{weak Poincaré inequality} an inequality of the
following form:
\begin{defn}
(Weak Poincaré inequality, $\alpha$-parameterization.) Given a Markov
transition operator $P$ on $\E$, we will say that\emph{ $P^{*}P$
satisfies a weak Poincaré inequality} if for any $f\in\ELL_{0}\left(\mu\right)$,
\begin{align*}
\|f\|_{2}^{2} & \leq\alpha(r)\mathcal{E}(P^{*}P,f)+r\Phi(f),\quad\forall r>0,
\end{align*}
where $\alpha:(0,\infty)\to[0,\infty)$ is a decreasing function,
and $\Phi:\ELL(\mu)\to[0,\infty]$ is a functional satisfying for
any $f\in\ELL(\mu)$, $c>0$ and $n\in\mathbb{N}$,
\begin{equation}
\Phi(cf)=c^{2}\Phi(f),\quad\Phi(P^{n}f)\le\Phi(f),\quad\|f-\mu(f)\|_{2}^{2}\leq a\Phi(f-\mu(f)),\label{eq:phi_condn}
\end{equation}
where $a:=\sup_{f\in\ELL_{0}(\mu)\backslash\{0\}}\|f\|_{2}^{2}/\Phi(f)$.\label{def:Weak-Poincar=0000E9-inequality}
\end{defn}

\begin{rem}
A popular choice of $\Phi$ is $\Phi=\|\cdot\|_{\osc}^{2}$, for which
\textcolor{black}{$a\leq1$,} but we will also later consider $\Phi=\|\cdot\|_{2p}^{2}$
for $p\ge1$, which also has \textcolor{red}{$a\leq1$} by Lyapunov's
inequality.

\end{rem}

\begin{rem}
\label{rem:strong_PI}Note that $\alpha(r)$ typically diverges as
$r\to0$. By contrast, a \emph{strong Poincaré inequality} refers
to the situation when $\alpha$ is uniformly bounded above by $\alpha(r)\le1/C_{\mathrm{P}}$
for some $C_{\mathrm{P}}>0$; in this case we may take $r\to0$ and
recover the standard strong Poincaré inequality $C_{\mathrm{P}}\|f\|_{2}^{2}\leq\mathcal{E}(P^{*}P,f)$
for $f\in\ELL_{0}(\mu)$, from which one can immediately deduce geometric
convergence \cite{fill-1991}, that is for any $f\in\ELL_{0}(\mu)$,
$n\in\mathbb{N}$,
\begin{equation}
\|P^{n}f\|_{2}^{2}\le(1-C_{\mathrm{P}})^{n}\|f\|_{2}^{2}.\label{eq:conv-with-SPI}
\end{equation}
In what follows we show that a weak Poincaré inequality implies the
existence of a function $n\mapsto\gamma(n)$, which is decreasing
to 0, such that for any $f\in\ELL_{0}(\mu)$ and $\Phi(f)<\infty$,
\begin{equation}
\|P^{n}f\|_{2}^{2}\le\gamma(n)\Phi(f).\label{eq:sub-geometric-convergence}
\end{equation}
\end{rem}

A very useful equivalent formulation of the weak Poincaré inequality,
\textcolor{black}{which bears some resemblance to the `super-Poincaré
inequality' of \cite{Rockner2001},} is the following.

\begin{defn}
(Weak Poincaré inequality, $\beta$-parameterization.) Given a Markov
transition operator $P$ on $\E$, we will say that \emph{$P^{*}P$
satisfies a weak Poincaré inequality} if for any $f\in\ELL_{0}(\mu)$,\label{def:Super-Poincar=0000E9-inequality}
\[
\|f\|_{2}^{2}\leq s\mathcal{E}(P^{*}P,f)+\beta(s)\Phi(f),\quad\forall s>0,
\]
where $\beta:(0,\infty)\to[0,\infty)$ is a decreasing function with
$\beta(s)\downarrow0$ as $s\to\infty$, and $\Phi:\ELL(\mu)\to[0,\infty]$
is a functional satisfying (\ref{eq:phi_condn}) for any $f\in\ELL(\mu)$,
$c>0$ and $n\in\mathbb{N}$.
\end{defn}

\textcolor{black}{These two formulations are equivalent; see our Remark~\ref{rem:wpi_equiv}
below, and we will typically refer to a `weak Poincaré inequality'
without specifying the paramterization. If there is ambiguity, we
will write $\alpha$- or $\beta$-weak Poincaré inequality to specify
the parameterization.} Because $a$ is such that $\|f\|_{2}^{2}\leq a\Phi(f)$
for all $f\in\ELL_{0}(\mu)$, one can always take $\beta\leq a$ in
Definition~\ref{def:Super-Poincar=0000E9-inequality} and $\alpha(r)=0$
for $r\geq a$ in Definition~\ref{def:Weak-Poincar=0000E9-inequality}.

\begin{rem}
\label{rem:wpi_equiv}Suppose an\textcolor{black}{{} $\alpha$-weak
}Poincaré inequality holds for a function $\alpha$ with $\alpha(r)=0$
for $r\geq a$. Then a\textcolor{black}{{} $\beta$-weak} Poincaré inequality
holds with $\beta(s):=\inf\{r>0:\alpha(r)\le s\}$. Conversely, suppose
a \textcolor{black}{$\beta$-weak} Poincaré inequality holds for a
function $\beta$ with $\beta\leq a$. Then an\textcolor{black}{{}
$\alpha$-weak} Poincaré inequality holds with $\alpha(r):=\inf\{s>0:\beta(s)\le r\}$.
This procedure always returns a right-continuous function, so for
a given $\alpha$ (or $\beta$) satisfying a weak Poincaré inequality,
iterating this procedure will return the right-continuous version
of $\alpha$ (or $\beta$).
\end{rem}

While in practice establishing a weak Poincaré inequality is often
the most tractable option, a third (essentially) equivalent formulation
plays an important rôle to establish (\ref{eq:sub-geometric-convergence})
with optimal rate function $\gamma$. We need the following functions:
\begin{defn}
\label{def:KandKstar}For $\beta$ as in Definition~\ref{def:Super-Poincar=0000E9-inequality}
we let 
\begin{enumerate}
\item $K\colon[0,\infty)\rightarrow[0,\infty)$ be such that $K(u):=u\,\beta(1/u)$
for $u>0$ and $K(0):=0$,
\item $K^{*}\colon[0,\infty)\rightarrow[0,\infty]$ be such that $K^{*}(v):=\sup_{u\ge0}\left\{ uv-K(u)\right\} $
is the convex conjugate of $K$.
\end{enumerate}
\end{defn}

Then for $f\in\ELL_{0}(\mu)$ such that $0<\Phi(f)<\infty$, the \textcolor{black}{weak
}Poincaré inequality can be formulated as follows with $u=1/s>0$,
\[
u\left\Vert f\right\Vert _{2}^{2}\leq\mathcal{E}\left(P^{*}P,f\right)+K\left(u\right)\Phi\left(f\right),
\]
which by rearranging terms and optimising leads to
\[
K^{*}\left(\frac{\left\Vert f\right\Vert _{2}^{2}}{\Phi(f)}\right)\leq\frac{\mathcal{E}(P^{*}P,f)}{\Phi(f)}.
\]
Relevant properties of $K^{*}$ can be found in Lemma~\ref{lem:properties-K-star}.
The rate function $\gamma$ in (\ref{eq:sub-geometric-convergence})
is the inverse function of $F_{a}$ given below, which is well defined:
\begin{lem}
\label{lem:F_a_properties}Let $F_{a}(\cdot)\colon(0,a]\rightarrow\mathbb{R}$,
where $(0,a]\subset\mathrm{D}$, be given by
\[
F_{a}(x):=\int_{x}^{a}\frac{{\rm d}v}{K^{*}(v)},
\]
where $K^{*}$ is given in Definition~\ref{def:KandKstar} and $\mathrm{D}:=\{v\ge0:K^{*}(v)<\infty\}$.
Then $F_{a}(\cdot)$
\begin{enumerate}
\item is well-defined, convex, continuous and strictly decreasing;
\item is such that $\lim_{x\downarrow0}F_{a}(x)=\infty$ ;
\item has a well-defined inverse function $F_{a}^{-1}:(0,\infty)\to(0,a)$,
with $F_{a}^{-1}(x)\to0$ as $x\to\infty$.
\end{enumerate}
\end{lem}

The main result of this section is as follows.
\begin{thm}
\label{thm:WPI_F_bd}Assume that $\mu$ and $P^{*}P$ satisfy a\textcolor{black}{{}
weak }Poincaré inequality as in Definition~\ref{def:Super-Poincar=0000E9-inequality}.
Then for $f\in\ELL_{0}(\mu)$ such that $0<\Phi(f)<\infty$  and
any $n\in\mathbb{N}$,
\[
\|P^{n}f\|_{2}^{2}\leq\Phi\left(f\right)F_{a}^{-1}\left(n\right),
\]
where $F_{a}\colon(0,a]\rightarrow\mathbb{R}$ is the decreasing convex
and invertible function as in Lemma~\ref{lem:F_a_properties}.
\end{thm}

\begin{rem}
\label{rem:F_=00005Cinfty}When $\int_{a}^{\infty}\frac{\dif v}{K^{*}\left(v\right)}<\infty$,
one can define $F_{\infty}(x):=\int_{x}^{\infty}\frac{{\rm d}v}{K^{*}\left(v\right)}$
for each $x>0$, and since $F_{a}\left(x\right)\le F_{\infty}\left(x\right)$,
one can similarly derive a bound $\|P^{n}f\|_{2}^{2}\leq\Phi\left(f\right)F_{\infty}^{-1}\left(n\right)$.
\end{rem}

\begin{rem}
A different proof relying on an alternative use of the Poincaré inequality
is given in Appendix~\ref{app:rockner-wang} for completeness, which
corresponds to the formulation of \cite[Theorem 2.1]{Rockner2001},
advocated by the authors for its tractability, but leads to suboptimal
results (see comments in Appendix~\ref{app:rockner-wang}). We have
found the formulation of Theorem~\ref{thm:WPI_F_bd} sufficiently
flexible for our applications.\textcolor{black}{{} This general approach
was in fact suggested in the continuous-time setting, see e.g. \cite[equation (1.4)]{Rockner2001}
but only later utilised in \cite{bakry}, where improved rates were
obtained. Our approach here can be seen as the natural discrete-time
analogue, however we further generalize the approach to allow for
general $\Phi$ and $a\neq\infty$, and make explicit the connection
with convex conjugates.}
\end{rem}

\begin{rem}
\label{rmk:strong-PI-super-beta}If $P^{*}P$ satisfies a strong Poincaré
inequality with constant $C_{{\rm P}}$, one may take the corresponding
$\beta$ to be $\beta(s)=a\mathbb{I}\{s\leq C_{{\rm P}}^{-1}\}$.
Conversely, if $\beta(s)=a\mathbb{I}\{s\leq C_{{\rm P}}^{-1}\}$ then
one can deduce that a strong Poincaré inequality holds. A simple calculation
shows that $K^{*}(v)=C_{{\rm P}}v$, for $0\leq v\leq a$, and
\[
F_{a}(x)=\int_{x}^{a}\frac{{\rm d}v}{C_{{\rm P}}v}{\rm d}v=C_{{\rm P}}^{-1}\log\left(\frac{a}{x}\right),
\]
from which we recover an exponential rate. However, $F_{a}^{-1}(n)=a\exp\big(-C_{{\rm P}}n\big)$
and since $\exp\big(-C_{{\rm P}}n\big)\geq(1-C_{{\rm P}})^{n}$ because
$-x/\sqrt{1-x}\leq\log(1-x)\leq-x$ for $x\in[0,1)$, this suggests
a loss compared to a more direct method leading to (\ref{eq:conv-with-SPI}).
In this setting we may also take $\Phi(f)=\|f\|_{2}^{2}$ and $a=1$.
\end{rem}

\begin{rem}
\label{rem:TV}It is possible to relate convergence in the sense of
Theorem~\ref{thm:WPI_F_bd} to $\ELL(\mu)$ convergence of $\nu(P^{*})^{n}$
and to convergence in total variation, for some initial distribution
$\nu$, where
\[
\left\Vert \nu(P^{*})^{n}-\mu\right\Vert _{2}:=\left\Vert \frac{{\rm d}(\nu(P^{*})^{n})}{{\rm d}\mu}-1\right\Vert _{2},
\]
bearing in mind that in the reversible case $P^{*}=P$. For any distribution
$\nu\ll\mu$ such that $\Phi\left(\frac{{\rm d}\nu}{{\rm d}\mu}-1\right)<\infty$,
we have 
\begin{align*}
\|\nu(P^{*})^{n}-\mu\|_{{\rm TV}} & =\int\left|\frac{{\rm d}(\nu(P^{*})^{n})}{{\rm d}\mu}(x)-1\right|{\rm d}\mu(x)\\
 & \leq\left\Vert \frac{{\rm d}(\nu(P^{*})^{n})}{{\rm d}\mu}-1\right\Vert _{2}\\
 & =\left\Vert P^{n}\left(\tfrac{{\rm d}\nu}{{\rm d}\mu}\right)-1\right\Vert _{2}\\
 & \leq\Phi\left(\frac{{\rm d}\nu}{{\rm d}\mu}-1\right)^{1/2}F_{a}^{-1}(n)^{1/2},
\end{align*}
where the first inequality follows from Jensen's inequality and the
last equality follows from \cite[p. 529]{Douc18book}. The condition
on $\nu$ is not very restrictive in practical settings when one can
choose $\nu$. For example, if $\mu$ and $\nu$ have densities with
respect to some common reference measure one can choose $\nu$ to
be uniformly distributed on some set on which $\mu$ has a positively
lower-bounded density.
\end{rem}

\begin{rem}
\textcolor{black}{By following the proof of Theorem~\ref{thm:WPI_F_bd}
and stopping early, one can obtain bounds which are tighter but sometimes
less convenient to work with. }

\textcolor{black}{For example, writing $T^{\circ n}$ for the $n$-fold
composition of the map $T$ with itself, one can obtain the bound
\begin{equation}
\frac{\left\Vert P^{n}f\right\Vert _{2}^{2}}{\Phi\left(f\right)}\leqslant\left(\Id-K^{*}\right)^{\circ n}\left(\frac{\left\Vert f\right\Vert _{2}^{2}}{\Phi\left(f\right)}\right),\label{eq:rmk_bd1}
\end{equation}
and indeed a decay estimate of this form is equivalent to the original
WPI holding with no loss of information; take $n=1$. Going one step
further in the proof, one can obtain the bound
\[
\frac{\left\Vert P^{n}f\right\Vert _{2}^{2}}{\Phi\left(f\right)}\leqslant F_{a}^{-1}\left(n+F_{a}\left(\frac{\left\Vert f\right\Vert _{2}^{2}}{\Phi\left(f\right)}\right)\right),
\]
which is weaker than (\ref{eq:rmk_bd1}) due to the integral approximation,
but stronger than the separable bound which is stated in the theorem.}
\end{rem}

A useful lemma we will make use of later concerning linear rescalings
is the following.
\begin{lem}
\label{lem:rescaling-beta-rescaling-invF}Let $\tilde{\beta}(s):=c_{1}\beta(c_{2}s)$
for $c_{1,}c_{2}>0$. Then $\tilde{K}^{*}(v):=\sup_{u\in\mathbb{R_{+}}}u[v-\tilde{\beta}(1/u)]=c_{1}c_{2}K^{*}(v/c_{1})$
and the corresponding function\textcolor{black}{{} $\tilde{F}_{a}(w)=c_{2}^{-1}F_{a/c_{1}}(w/c_{1})$.
Furthermore, when $c_{1}\ge1$, $\tilde{F}_{a}(w)\le c_{2}^{-1}F_{a}(w/c_{1})$,
and we can conclude $\tilde{F}_{a}^{-1}(x)\le c_{1}F_{a}^{-1}(c_{2}x)$.}
\end{lem}

\subsubsection{Examples of $\beta\left(s\right)$ and $\gamma=F_{a}^{-1}$\label{subsec:Examples-of-b}}

Throughout the following examples\textcolor{black}{{} (which coincide
with those of \cite[Corollary~2.4]{Rockner2001})}, we use the notation
of Theorem~\ref{thm:WPI_F_bd}.
\begin{lem}
\label{lem:rate-beta-decays-polynomial}For $\beta(s)=c_{0}s^{-c_{1}}$,
$K^{*}(v)=C\left(c_{0},c_{1}\right)v^{1+c_{1}^{-1}}$. Then with $F_{\infty}$
as in Remark~\ref{rem:F_=00005Cinfty}, the convergence rate is bounded
by
\[
F_{\infty}^{-1}\left(n\right)\le c_{0}\left(1+c_{1}\right)^{1+c_{1}}n^{-c_{1}}.
\]
\end{lem}

\begin{lem}
\label{lem:rate-beta-decays-exponentially}Assume $\beta(s)=\eta_{0}\exp\left(-\eta_{1}s^{\eta_{2}}\right)$
for $\eta_{0},\eta_{1},\eta_{2}>0$ and choose $a>0$. Then there
exist $C>0$, $0<v_{0}<1\wedge a$ such that for $v\in[0,v_{0}]$,
\[
K^{*}\left(v\right)\geqslant Cv\left(\log\left(\frac{1}{v}\right)\right)^{-1/\eta_{2}}.
\]
 In addition, there exists $C'>0$ such that for all $n\in\mathbb{N}$,
\[
F_{a}^{-1}(n)\leq C'\exp\left(-\left(C\frac{1+\eta_{2}}{\eta_{2}}n\right)^{\eta_{2}/\left(1+\eta_{2}\right)}\right).
\]
\end{lem}

\textcolor{red}{}
\begin{lem}
\textcolor{red}{\label{lem:rate-beta-decays-log}}\textcolor{black}{Assume
$\beta\left(s\right)=c_{0}\cdot\left(\log\max\left(c_{1},s\right)\right)^{-p}$
for $c_{0}>0$, $c_{1}>1$, $p>0$. Then there exist $v_{0}>0$, $C>0$,
such that for $v\in\left[0,v_{0}\right]$,}

\textcolor{black}{
\[
K^{*}\left(v\right)\geqslant C\cdot v^{1+1/p}\cdot\exp\left(-\left(v/c_{0}\right)^{-1/p}\right).
\]
In addition, there exists $C'>0$ such that for all $n\in\mathbb{N}$,}

\textcolor{black}{
\[
F_{a}^{-1}\left(n\right)\leqslant C'\cdot\left(\log\max\left(n,2\right)\right)^{-p}.
\]
}
\end{lem}

\subsection{Reversible case\label{subsec:Reversible-case}}

When the kernel $P$ is reversible with respect to $\mu$, we can
derive a simplified weak Poincaré inequality in terms of $P$ directly,
rather than $P^{*}P$, making the approach much more practical. This
kind of result seems to be new to the best of our knowledge, \textcolor{black}{and
indeed the need to handle $P^{*}P$ is one of the key subtleties of
our present discrete-time setting as opposed to the continuous-time
setting.} Furthermore we can also derive a converse result; a weak
Poincaré inequality is necessary for subgeometric convergence.

\subsubsection{Simplified weak Poincaré inequality\label{subsec:Simplified-weak-Poincar=0000E9}}
\begin{defn}
\label{def:WPI_rev}(Weak Poincaré inequality; reversible case.) Given
a reversible Markov transition operator $P$ on $\E$, we will say
that \emph{$P$ satisfies a weak Poincaré inequality} if for any $f\in\ELL_{0}\left(\mu\right)$,
\[
\left\Vert f\right\Vert _{2}^{2}\leq\alpha\left(r\right)\mathcal{E}\left(P,f\right)+r\Phi\left(f\right),\quad\forall r>0,
\]
where $\alpha:\left(0,\infty\right)\to[0,\infty)$ is a decreasing
function, and $\Phi:\ELL\left(\mu\right)\to\left[0,\infty\right]$
is a functional satisfying: for any $f\in\ELL\left(\mu\right)$, $c>0$
and $n\in\mathbb{N}$,
\[
\Phi\left(cf\right)=c^{2}\Phi\left(f\right),\quad\Phi\left(P^{n}f\right)\le\Phi\left(f\right),\quad\|f-\mu\left(f\right)\|_{2}^{2}\leq a\Phi\left(f-\mu\left(f\right)\right),
\]
with $a:=\sup_{f\in\ELL_{0}(\mu)\backslash\{0\}}\left\Vert f\right\Vert _{2}^{2}/\Phi\left(f\right)$.

A \textit{\textcolor{red}{$\beta$-weak}}\textit{\textcolor{red}{\emph{
}}}\emph{Poincaré inequality for $P$} is analogously defined: for
a function $\beta:\left(0,\infty\right)\to[0,\infty)$ decreasing
with $\beta\left(s\right)\downarrow0$ as $s\to\infty$, for any $f\in\ELL_{0}\left(\mu\right)$,

\[
\left\Vert f\right\Vert _{2}^{2}\leq s\mathcal{E}\left(P,f\right)+\beta\left(s\right)\Phi\left(f\right),\quad\forall s>0.
\]
\end{defn}

We are interested now to obtain an appropriate Poincaré inequality
for $P^{*}P=P^{2}$ from a corresponding Poincaré inequality for $P$.
The key complication is the left-hand side of the spectrum, around
$-1$. In order to rule out periodic behaviour (which will prevent
convergence), some assumptions on the spectrum in a neighbourhood
of $-1$ are required.
\begin{lem}
\label{lem:P-to-P2-gap}Suppose that the reversible kernel $P$ possesses
a left spectral gap: there exists some $0<c_{\mathrm{gap}}\le1$ such
that the spectrum of $P$ is bounded below: 
\[
\inf\sigma\left(P\right)\ge-1+c_{\mathrm{gap}}.
\]
Then we obtain the bound on the Dirichlet forms given $f\in\ELL\left(\mu\right)$
by
\[
\calE\left(P^{2},f\right)\ge c_{\mathrm{gap}}\calE\left(P,f\right).
\]
\end{lem}

\begin{cor}
\label{cor:rev_P2}In the setting of Lemma~\ref{lem:P-to-P2-gap},
it immediately follows that if $P$ satisfies a \textcolor{black}{weak}
Poincaré inequality with function $\beta$ as in Definition~\ref{def:WPI_rev},
$P^{2}$ satisfies a \textcolor{black}{weak }Poincaré inequality with
$\tilde{\beta}$ given by $\tilde{\beta}\left(s\right):=\beta\left(c_{\mathrm{gap}}s\right)$:
\[
\left\Vert f\right\Vert _{2}^{2}\leq s\mathcal{E}\left(P^{2},f\right)+\tilde{\beta}\left(s\right)\Phi\left(f\right)\quad,\forall s>0.
\]
Thus the convergence rate $\tilde{F}_{a}^{-1}$ for $\left\Vert P^{n}f\right\Vert _{2}^{2}$
can be immediately deduced from Theorem~\ref{thm:WPI_F_bd} and Lemma~\ref{lem:rescaling-beta-rescaling-invF}.
\end{cor}

\textcolor{black}{When there is no left spectral gap, we can generalize
the above results using a weak Poincaré inequality for $-P$.}
\begin{thm}
\textcolor{black}{\label{prop:left-WPI-P}Suppose $P$ is $\mu$-reversible.
Assume the following two weak-Poincaré inequalities hold: for all
$s>0$, $f\in\mathrm{L}_{0}^{2}(\mu)$,
\begin{align}
\|f\|_{2}^{2} & \le s\calE(-P,f)+\beta_{-}(s)\Phi(f)\label{eq:beta-}\\
\|f\|_{2}^{2} & \le s\calE(P,f)+\beta_{+}(s)\Phi(f).\label{eq:beta+}
\end{align}
Then, the following weak Poincaré inequality for $P^{2}$ holds: 
\begin{equation}
\|f\|_{2}^{2}\le s\calE(P^{2},f)+\beta(s)\tilde{\Phi}(f),\label{eq:P2_WPI_beta}
\end{equation}
for all $s>0,$ $f\in\mathrm{L}_{0}^{2}(\mu)$, where
\[
\beta(s):=\inf\{s_{1}\beta_{+}(s_{2})+\beta_{-}(s_{1})|s_{1}>0,s_{2}>0,s_{1}s_{2}=s\},
\]
\[
\tilde{\Phi}(f):=\Phi(f)\vee\Phi\left((\Id+P)^{1/2}f\right).
\]
}
\end{thm}

Recall that a $\mu$-reversible Markov kernel $P$ is \emph{positive}
if for any $f\in\ELL(\mu)$, $\langle Pf,f\rangle\ge0$, and a positive
reversible kernel $P$ has spectrum contained in the nonnegative interval
$\sigma(P)\subset[0,1]$. When $P$ is reversible and positive, convergence
of $P^{n}$ can be straightforwardly derived as then $c_{{\rm gap}}=1$. 
\begin{thm}
\label{thm:rev_pos_conv}Assume that the kernel $P$ is reversible
and positive and satisfies a \textcolor{black}{weak} Poincaré inequality
as in Definition~\ref{def:WPI_rev}. Then Theorem~\ref{thm:WPI_F_bd}
applies, so for $f\in\ELL_{0}(\mu)$ such that $0<\Phi(f)<\infty$
and any $n\in\mathbb{N}$,
\[
\left\Vert P^{n}f\right\Vert _{2}^{2}\leq\Phi\left(f\right)F_{a}^{-1}\left(n\right).
\]
\end{thm}

\begin{proof}
Since $P$ is reversible and positive, we can apply Corollary~\ref{cor:rev_P2}
with $c_{\mathrm{gap}}=1$ to see that $P^{2}$ satisfies a \textcolor{black}{weak}
Poincaré inequality with the same function $\beta$. We can then immediately
apply Theorem~\ref{thm:WPI_F_bd} to conclude.
\end{proof}
\begin{rem}
Realistic MCMC kernels will all possess a non-zero left spectral gap.
Indeed, popular methods such as the Independent Metropolis--Hastings
sampler, many random walk Metropolis algorithms, and the resulting
pseudo-marginal chains we will consider later are even positive \cite[Proposition 16]{Andrieu15}.
Furthermore, a given reversible kernel $P$ can be straightforwardly
modified to possess a positive left spectral gap by considering the
so-called lazy chain $Q:=\epsilon\Id+(1-\epsilon)P$ for $\epsilon\in[0,1)$.
Indeed, one of our contributions in Section~\ref{sec:Chaining-Poincar=0000E9-inequalities}
is to generalize this construction and give versions of Lemma~\ref{lem:P-to-P2-gap}
and Corollary~\ref{cor:rev_P2} holding under weaker assumptions;
see Theorem~\ref{thm:weakly-lazy-chain}.
\end{rem}

\subsubsection{Necessity of\textcolor{black}{{} weak }Poincaré inequalities}

We can also derive a converse to Theorem~\ref{thm:WPI_F_bd} in the
reversible setting. For our explicit examples of $\beta$ in Section~\ref{subsec:Examples-of-b}
as well as in the geometric case (Remark~\ref{rem:strong_PI}), it
turns out that we can derive a converse result, thus demonstrating,
at least in the reversible setting, that our approach is able to recover
the best possible rates of convergence for a given $\beta$ when $\beta$
is polynomial or polylogarithmic; see Remark~\ref{rem:nec-comment}.\textcolor{black}{{}
In the continuous-time setting, similar converse results have also
been derived; see, for instance \cite[Theorem 2.3]{Rockner2001}.}
\begin{prop}
\label{prop:necessity-reversible-scenario}Let $P$ be $\mu-$self-adjoint
Markov transition operator and assume for any $n\in\mathbb{N}$, and
$f\in\ELL_{0}(\mu)$,
\begin{equation}
\left\Vert P^{n}f\right\Vert _{2}^{2}\le\gamma\left(n\right)\Phi\left(f\right),\label{eq:nec_cond}
\end{equation}
for some functional $\Phi:\ELL\left(\mu\right)\to\left[0,\infty\right]$
satisfying (\ref{eq:phi_condn}) and a decreasing function $\gamma:\R_{+}\to\left(0,\infty\right)$,
with $\gamma\left(s\right)\to0$ as $s\to\infty$. Then $P^{2}$ satisfies
a \textcolor{black}{weak} Poincaré inequality (Definition~\ref{def:Super-Poincar=0000E9-inequality})
with
\[
\beta\left(s\right)\leqslant\beta_{1}\left(s\right):=\sup_{t\geqslant s}\inf_{n\geq2}\left\{ \frac{t^{n}}{\left(t-1\right)^{n-1}}\cdot\frac{\left(n-1\right)^{n-1}}{n^{n}}\cdot\gamma\left(n\right)\right\} ,\quad\forall s>1,
\]
which is decreasing and decreases to $0$.

\textcolor{black}{Similarly, suppose that under the same assumptions
on $\left(\mu,f,\Phi\right)$, and only assuming $P$ to be $\mu$-invariant,
it holds that
\[
\left\Vert P^{n}f\right\Vert _{2}^{2}\leqslant\Phi\left(f\right)\cdot F^{-1}\left(n+F\left(\frac{\left\Vert f\right\Vert _{2}^{2}}{\Phi\left(f\right)}\right)\right)
\]
for a function $F:\R_{+}\to\left(0,\infty\right)$ which is decreasing,
continuous, divergent at $0$, with an inverse function $F^{-1}$
which is decreasing, continuous, and convex, and such that $\log\left(-\mathrm{D}F^{-1}\right)$
is convex.}

\textcolor{black}{Then $P^{*}P$ satisfies a weak Poincaré inequality
with
\[
K^{*}\left(v\right)\leqslant K_{1}^{*}\left(v\right)=v-F^{-1}\left(1+F\left(v\right)\right),
\]
from which a corresponding $\beta$ can be deduced via convex duality.}

\textcolor{black}{Finally, suppose that under the same assumptions
on $\left(\mu,f,\Phi\right)$, and again only assuming $P$ to be
$\mu$-invariant, it holds that
\[
\frac{\left\Vert P^{n}f\right\Vert _{2}^{2}}{\Phi\left(f\right)}\leqslant\left(\Id-\tilde{K}^{*}\right)^{\circ n}\left(\frac{\left\Vert f\right\Vert _{2}^{2}}{\Phi\left(f\right)}\right),
\]
for some function $\tilde{K}^{*}:\left[0,a\right]\to\left[0,a\right]$
which is increasing, convex, and vanishes at $0$. Then $P^{*}P$
satisfies a weak Poincaré inequality with
\[
K^{*}\left(v\right)\leqslant\tilde{K}^{*}\left(v\right),
\]
from which a corresponding $\beta$ can again be deduced via convex
duality.}
\end{prop}

\begin{rem}
\label{rem:nec-comment}For explicit computations, it can be useful
to apply the elementary bounds 
\[
\frac{s^{n}}{\left(s-1\right)^{n-1}}=\left(s-1\right)\cdot\left(\frac{s}{s-1}\right)^{n}\leqslant\left(s-1\right)\cdot\exp\left(\frac{n}{s-1}\right),
\]
and 
\[
\frac{\left(n-1\right)^{n-1}}{n^{n}}=\frac{1}{n}\cdot\left(\frac{n-1}{n}\right)^{n-1}\leqslant\frac{1}{2n},
\]
to bound
\begin{multline*}
\inf_{n\geq2}\left\{ \frac{s^{n}}{\left(s-1\right)^{n-1}}\cdot\frac{\left(n-1\right)^{n-1}}{n^{n}}\cdot\gamma\left(n\right)\right\} \\
\leqslant\frac{1}{2}\cdot\inf_{n\geq2}\left\{ \gamma\left(n\right)\left(\frac{n}{s-1}\right)^{-1}\cdot\exp\left(\frac{n}{s-1}\right)\right\} \\
\leqslant\frac{1}{2}\cdot\sup_{t\geqslant s}\inf_{n\geq2}\left\{ \gamma\left(n\right)\left(\frac{n}{t-1}\right)^{-1}\cdot\exp\left(\frac{n}{t-1}\right)\right\} =:\beta_{2}\left(s\right),
\end{multline*}
which is often more convenient to work with, and is again decreasing
and decreases to $0$.

\textcolor{black}{We may consider the following procedure. Given
\[
\left\Vert f\right\Vert _{2}^{2}\leq s\mathcal{E}\left(P^{*}P,f\right)+\beta\left(s\right)\Phi\left(f\right),\quad\forall s>0,
\]
apply our Theorem~\ref{thm:WPI_F_bd} to deduce that $\left\Vert P^{n}f\right\Vert _{2}^{2}\leqslant\Phi\left(f\right)\cdot\gamma\left(n\right)$.
Then, apply the above construction to show that $P^{*}P$ satisfies
\[
\left\Vert f\right\Vert _{2}^{2}\leq s\mathcal{E}\left(P^{*}P,f\right)+\beta_{2}\left(s\right)\Phi\left(f\right),\quad\forall s>0,
\]
with $\beta_{2}$ as above. In the examples considered in Section
\ref{subsec:Examples-of-b}, we find that $\beta_{2}\left(s\right)\leqslant c_{1}\cdot\beta\left(c_{2}\cdot s\right)$
for positive constants $c_{1}$, $c_{2}$; see the Appendix for explicit
calculations. We do not address whether this will hold in general.
In particular, if $\beta$ is polynomial or polylogarithmic then $\beta_{2}$
and $\beta$ have the same asymptotic behaviour up to a multiplicative
constant. }

\end{rem}

\subsection{Illustration: Independent MH sampler\label{subsec:Illustration:-Independent-MH}}

As a concrete illustration of the results of Subsection\ \ref{subsec:Reversible-case}
we consider the Independent Metropolis--Hastings (IMH) algorithm.
This has been studied previously in \cite{Jarner02,Jarner07} using
drift/minorization conditions, and we show in this subsection that
we recover the same subgeometric rates of convergence using \textcolor{black}{weak}
Poincaré inequalities. We fix a target density $\pi$ and a positive
proposal density $q$ on $\E$ and define
\[
w(x):=\frac{\pi(x)}{q(x)},\quad x\in\E.
\]
Then the IMH chain has reversible transition kernel $P$ given by
\[
P(x,\dif y)=a(x,y)q(y)\dif y+\rho(x)\delta_{x}(\dif y),
\]
where $a(x,y)=\left[1\wedge\frac{w(y)}{w(x)}\right]$, and $\rho(x)=\int\left[1-a(x,y)\right]\,q(\dif y)$.
In this case, it follows from reversibility that we have the following
well-known representation:
\begin{lem}
We can express
\[
\mathcal{E}(P,f)=\frac{1}{2}\int\pi(x)\pi(y)\left(w^{-1}(x)\wedge w^{-1}(y)\right)\left[f(y)-f(x)\right]^{2}\,\dif x\,\dif y,
\]
and
\[
\|f\|_{2}^{2}=\frac{1}{2}\int\pi(x)\pi(y)\left[f(y)-f(x)\right]^{2}\,\dif x\,\dif y.
\]
\end{lem}

For the IMH, the following is known \cite{Jarner02,Jarner07}:
\begin{prop}
If $w$ is uniformly bounded from above, then the IMH sampler is uniformly
ergodic. However, if $w$ is not uniformly bounded above, then the
chain is not even geometrically ergodic.
\end{prop}

Thus since we are interested in the case of subgeometric convergence,
we assume that $w$ is not bounded from above, or equivalently, $w^{-1}$
is not bounded from below by any positive constant. Thus given any
$s>0$, we define the following sets:
\[
A(s):=\left\{ (x,y)\in\E\times\E:w^{-1}(x)\wedge w^{-1}(y)\ge1/s\right\} .
\]
Since we are assuming the subgeometric case, there is no $s>0$ for
which $A(s)=\E\times\E$.
\begin{prop}
\label{prop:IMH-WPI}For the IMH, we have the following\textcolor{black}{{}
weak }Poincaré inequality: given any $f\in\ELL_{0}(\pi)$ and $s>0$,
\[
\|f\|_{2}^{2}\le s\mathcal{E}(P,f)+\frac{\pi\otimes\pi(A(s)^{\complement})}{2}\|f\|_{\mathrm{osc}}^{2}.
\]
\end{prop}

Our bound in Proposition~\ref{prop:IMH-WPI} allows us to directly
link the tail properties of the weights $w(x)$ under $\pi$ and the
resulting rates of subgeometric convergence. We can apply Theorem~\ref{thm:rev_pos_conv},
since the IMH kernel is always positive \cite{Gasemyr06}.
\begin{rem}
Using this approach we recover convergence rates identical to those
obtained using drift/minorization conditions \cite{Jarner02,Douc2007}.
Notice that $\pi\otimes\pi(A(s)^{\complement})\leq2\pi\big(w(X)\geq s\big)$,
suggesting the use of Markov's inequality followed by the application
of Lemmas~\ref{lem:rate-beta-decays-polynomial}--\ref{lem:rate-beta-decays-exponentially}
to obtain upper bounds on the rate of convergence of the IMH algorithm
for the total variation distance (Remark~\ref{rem:TV}). In the scenario
where $\mathbb{E}_{\pi}\left[w(X)^{\eta}\right]<\infty$ with $\eta>1$
it is possible to determine a drift condition for the IMH algorithm
(see for example \cite[Theorem 5.3]{Jarner02}, \cite[Proposition 9]{Jarner07}),
and deduce a rate of $n^{-\eta}$. This is the same as can be obtained
using our Lemma~\ref{lem:rate-beta-decays-polynomial}; see also
our examples below. Similarly, when $\mathbb{E}_{\pi}\left[\exp\left(w(X)^{\eta}\right)\right]<\infty$
it is possible to identify a drift condition \cite[Lemma 56]{Andrieu15}
for the algorithm and deduce the existence of $C_{1},C_{2}>0$ such
that the total variation convergence rate is bounded \cite[p. 1365]{Douc2004}
by
\[
C_{1}n^{1/(1+\eta)}\exp\left(-\left\{ C_{2}\frac{1+\eta}{\eta}n\right\} {}^{\eta/(1+\eta)}\right),
\]
which is consistent with our result in Lemma~\ref{lem:rate-beta-decays-exponentially},
except for the polynomial term. However we note that it is not possible
to derive the precise value of the constant $C_{2}$ and cannot therefore
conclude which of the two approaches provides the fastest rate. 
\end{rem}

We turn now to some concrete examples inspired by \cite{Jarner02,Jarner07}
where $\beta(s)$ can be evaluated directly.

\subsubsection{Exponential target and proposal case}

We work on $\E=(0,\infty)\subset\R$, and have target and proposal
densities
\[
\pi(x)=a_{1}\exp(-a_{1}x),\quad q(x)=a_{2}\exp(-a_{2}x).
\]
Since we are interested in the subgeometric case, we assume that $a_{2}>a_{1}.$
For this example, it was shown in \cite[Proposition 9(b)]{Jarner07}
that there is polynomial convergence, with rate at least $a_{1}/(a_{2}-a_{1})$.
\begin{lem}
\label{lem:indep_ex}We have that for $s\geq1$,
\[
\frac{\pi\otimes\pi(A(s)^{\complement})}{2}=\frac{1}{2}\left[1-\left(1-s^{-\frac{a_{1}}{a_{2}-a_{1}}}\right)^{2}\right]\leq s^{^{-\frac{a_{1}}{a_{2}-a_{1}}}}.
\]
\end{lem}

In this case, we can make use of Lemma~\ref{lem:rate-beta-decays-polynomial}
to conclude the following, consistent with \cite[Proposition 9(b)]{Jarner07}.
\begin{prop}
For our exponential example, we recover the convergence rate for some
$C>0$,
\[
\|P^{n}f\|_{2}^{2}\le C\|f\|_{\osc}^{2}n^{-\frac{a_{1}}{a_{2}-a_{1}}}.
\]
\end{prop}

\subsubsection{Polynomial target and proposal case}

We take $\E=[1,\infty)\subset\R$, and target and proposal densities
\[
\pi(x)=\frac{b_{1}}{x^{1+b_{1}}},\quad q(x)=\frac{b_{2}}{x^{1+b_{2}}}.
\]
We are interested in subgeometric convergence, so assume that $b_{2}>b_{1}.$
It was shown in \cite[Proposition 9(a)]{Jarner07} that for this example
there is polynomial convergence with rate at least $b_{1}/(b_{2}-b_{1})$.
An entirely analogous calculation to the exponential example above
allows us to conclude:
\begin{prop}
For our polynomial example, we obtain the convergence rate, for some
$C>0,$
\[
\|P^{n}f\|_{2}^{2}\le C\|f\|_{\osc}^{2}n^{-\frac{b_{1}}{b_{2}-b_{1}}}.
\]
\end{prop}

\section{Chaining Poincaré inequalities\label{sec:Chaining-Poincar=0000E9-inequalities}}

In this section we show how comparison of Dirichlet forms can be used
to deduce convergence properties of a given Markov chain from another
one. These results extend existing quantitative comparison results.
\begin{prop}
\label{prop:chaining-with-SPI}Let $P_{1}$ and $P_{2}$ be two $\mu-$invariant
Markov kernels. Let $T_{i}=P_{i}$ or $T_{i}=P_{i}^{*}P_{i}$. Assume
that for all $s>0$ and $f\in\ELL_{0}(\mu)$,
\begin{align*}
C_{\mathrm{P}}\|f\|_{2}^{2} & \leq\mathcal{E}(T_{1},f)\\
\mathcal{E}(T_{1},f) & \leq s\mathcal{E}(T_{2},f)+\beta'(s)\Phi(f),
\end{align*}
where
\begin{enumerate}
\item $C_{\mathrm{P}}>0$ and $\beta'\colon(0,\infty)\rightarrow(0,\infty)$
is decreasing and $\beta'(s)\downarrow0$ as $s\to\infty$,
\item $\Phi\colon\ELL(\mu)\rightarrow[0,\infty]$ is such that (\ref{eq:phi_condn})
holds.
\end{enumerate}
Then for any $s>0$,
\[
\|f\|_{2}^{2}\le s\calE(T_{2},f)+\beta(s)\Phi(f),
\]
with $\beta(s)=\beta'(C_{\mathrm{P}}s)/C_{{\rm P}}$.
\end{prop}

The proof is immediate. 
\begin{rem}
This generalises the comparison of Dirichlet forms used in \cite{caracciolo1990}
which corresponds to $\beta(s)=0$ for all $s>\bar{s}$ for some $\bar{s}>0$.
Further, assume that for any $(x,A)\in\mathsf{E}\times\mathcal{F}$,
$P_{2}(x,A\setminus\{x\})\geq\varepsilon(x)P_{1}(x,A\setminus\{x\})$
for some $\varepsilon\colon\mathsf{E}\rightarrow(0,1]$, then with
$\left\{ (x,y)\in\mathsf{E}^{2}\colon\varepsilon(x)s>1\right\} $
and $s>0$ we have
\begin{align*}
\mathcal{E}(P_{1},f) & \leq\frac{1}{2}\int_{A(s)}\varepsilon(x)s\mu({\rm d}x)P_{1}(x,{\rm d}y)\left[f(y)-f(x)\right]^{2}\\
 & \hspace{4cm}+\frac{1}{2}\int_{A(s)^{\complement}}\mu({\rm d}x)P_{1}(x,{\rm d}y)\left[f(y)-f(x)\right]^{2}\\
 & \leq s\mathcal{E}(P_{2},f)+\frac{1}{2}\mu\left(\varepsilon^{-1}(X)\geq s\right)\|f\|_{{\rm osc}}^{2},
\end{align*}
which is a generalisation of \cite[Theorem A3]{caracciolo1990} and
together with Theorem~\ref{thm:rev_pos_conv} leads to a counterpart
of \cite[Theorem A1]{caracciolo1990} for rates of convergence. However,
we have not found an elegant generalisation of \cite[Theorem A2]{caracciolo1990}
concerned with asymptotic variances. Theorem~\ref{thm:practical-comparison-P1-P2}
further generalizes these comparison ideas.
\end{rem}

This can be further extended to the scenario where $T_{1}$ satisfies
a weak Poincaré inequality.
\begin{thm}
\label{thm:chaining-with-WPI}Let $P_{1}$ and $P_{2}$ be two $\mu-$invariant
Markov kernels. Let $T_{i}=P_{i}$ or $T_{i}=P_{i}^{*}P_{i}$. Assume
that for all $s>0$ and $f\in\ELL_{0}(\mu)$,
\begin{align}
\|f\|_{2}^{2} & \leq s\mathcal{E}(T_{1},f)+\beta_{1}(s)\Phi_{1}(f)\nonumber \\
\mathcal{E}(T_{1},f) & \leq s\mathcal{E}(T_{2},f)+\beta_{2}(s)\Phi_{2}(f),\label{eq:nested_WPI}
\end{align}
where
\begin{enumerate}
\item $\beta_{1},\beta_{2}\colon(0,\infty)\rightarrow(0,\infty)$ are decreasing
and $\beta_{1}(s),\beta_{2}(s)\downarrow0$ as $s\to\infty$,
\item $\Phi_{1},\Phi_{2}\colon\ELL(\mu)\rightarrow[0,\infty]$ are such
that (\ref{eq:phi_condn}) hold for $P_{1}$ and $P_{2}$ respectively,
\item for any $n\in\mathbb{N}$ and $f\in\ELL_{0}(\mu)$, $\Phi_{1}(P_{2}^{n}f)\leq\Phi_{1}(f)$.
\end{enumerate}
Then for any $s>0$,
\begin{equation}
\|f\|_{2}^{2}\le s\calE(T_{2},f)+\beta(s)\Phi(f),\label{eq:nested_WPI2}
\end{equation}
where $\Phi:=\Phi_{1}\vee\Phi_{2}$ and
\[
\beta(s):=\inf\left\{ s_{1}\beta_{2}(s_{2})+\beta_{1}(s_{1})\,|s_{1}>0,s_{2}>0,s_{1}s_{2}=s\right\} .
\]
Furthermore, $\beta:(0,\infty)\to(0,\infty)$ is monotone decreasing
and satisfies $\beta(s)\downarrow0$ as $s\to\infty$, $\Phi(cf)=c^{2}\Phi(f)$
for $c>0$ and $\Phi(P_{2}^{n}f)\leq\Phi(f)$ for any $n\in\mathbb{N}$
and $f\in\ELL_{0}(\mu)$.

Additionally, writing $K_{i}\left(u\right)=u\cdot\beta_{i}\left(1/u\right)$,
$K\left(u\right)=u\cdot\beta\left(1/u\right)$, it holds that
\begin{align*}
K\left(u\right) & =\inf\left\{ K_{2}(u_{2})+u_{2}K_{1}(u_{1})\,|u_{1}>0,u_{2}>0,u_{1}u_{2}=u\right\} ,\\
K^{*}\left(v\right) & =K_{2}^{*}\circ K_{1}^{*}\left(v\right).
\end{align*}

\end{thm}

\begin{proof}
Fix $s>0$. Given any $s_{1},s_{2}>0$ with $s_{1}s_{2}=s$, by direct
substitution in (\ref{eq:nested_WPI}), we can arrive at
\begin{align*}
\|f\|_{2}^{2} & \le s\calE(T_{2},f)+\beta_{1}(s_{1})\Phi_{1}(f)+s_{1}\beta_{2}(s_{2})\Phi_{2}(f)\\
 & \le s\calE(T_{2},f)+\left[\beta_{1}(s_{1})+s_{1}\beta_{2}(s_{2})\right]\left[\Phi_{1}(f)\vee\Phi_{2}(f)\right].
\end{align*}
Taking an infimum, we arrive at (\ref{eq:nested_WPI2}).

Now we prove the monotonicity of $\beta$. Fix some $s>0$ and any
$s_{1},s_{2}>0$ with $s_{1}s_{2}=s$. Given any $s'\ge s$, note
that
\begin{align*}
\beta(s') & \le s_{1}\beta_{2}(s'/s_{1})+\beta_{1}(s_{1})\\
 & \le s_{1}\beta_{2}(s/s_{1})+\beta_{1}(s_{1})\\
 & =s_{1}\beta_{2}(s_{2})+\beta_{1}(s_{1}).
\end{align*}
Here we made use of the fact that $\beta_{2}$ is a decreasing function.
Taking an infimum over $s_{1},s_{2}$, we conclude that $\beta(s')\le\beta(s)$.

We now show that given $\epsilon>0$, we can find $s>0$ such that
$\beta(s)\le\epsilon$. Combined with monotonicity, this proves that
$\beta(s)\downarrow0$ as $s\to\infty$. So fix $\epsilon>0$. Choose
$s_{1}>0$ such that $\beta_{1}(s_{1})\le\epsilon/2$, which can be
done since $\beta_{1}(s)\downarrow0$ as $s\to\infty$. Given such
an $s_{1}$, now choose $s_{2}>0$ large enough so that $s_{1}\beta_{2}(s_{2})\le\epsilon/2$.
Thus for $s:=s_{1}s_{2}$ for these choices of $s_{1},s_{2}$, we
have shown that $\beta(s)\le\epsilon/2+\epsilon/2=\epsilon$.

To complete the proof, write
\begin{align*}
K\left(u\right) & =u\cdot\beta\left(1/u\right)\\
 & =u\cdot\inf\left\{ s_{1}\beta_{2}(s_{2})+\beta_{1}(s_{1})\,|s_{1}>0,s_{2}>0,s_{1}s_{2}=1/u\right\} \\
 & =u\cdot\inf\left\{ \left(1/u_{1}\right)\cdot\beta_{2}(1/u_{2})+\beta_{1}(1/u_{1})\,|1/u_{1}>0,1/u_{2}>0,\left(1/u_{1}\right)\cdot\left(1/u_{2}\right)=1/u\right\} \\
 & =\inf\left\{ \left(u/u_{1}\right)\cdot\beta_{2}(1/u_{2})+\left(u_{1}u_{2}\right)\cdot\beta_{1}(1/u_{1})\,|u_{1}>0,u_{2}>0,u_{1}u_{2}=u\right\} \\
 & =\inf\left\{ u_{2}\beta_{2}(1/u_{2})+u_{2}\cdot u_{1}\beta_{1}(1/u_{1})\,|u_{1}>0,u_{2}>0,u_{1}u_{2}=u\right\} \\
 & =\inf\left\{ K_{2}\left(u_{2}\right)+u_{2}\cdot K_{1}\left(u_{1}\right)\,|u_{1}>0,u_{2}>0,u_{1}u_{2}=u\right\} ,
\end{align*}
as claimed. Finally, 
\begin{align*}
K^{*}\left(v\right) & :=\sup_{u\geqslant0}\left\{ uv-K\left(u\right)\right\} \\
 & =\sup_{u\geqslant0}\left\{ uv-\inf_{u_{1},u_{2}}\left\{ K_{2}\left(u_{2}\right)+u_{2}\cdot K_{1}\left(u_{1}\right)\right\} \right\} \\
 & =\sup_{u,u_{1},u_{2}}\left\{ uv-K_{2}\left(u_{2}\right)-u_{2}\cdot K_{1}\left(u_{1}\right)\right\} ,
\end{align*}
where $u_{1,}u_{2}$ are again constrained to be non-negative and
have their product equal to $u$. Now, rewrite $u=u_{1}u_{2}$ and
eliminate the variable $u$ to write
\begin{align*}
K^{*}\left(v\right) & =\sup_{u_{1},u_{2}>0}\left\{ u_{1}u_{2}v-K_{2}\left(u_{2}\right)-u_{2}\cdot K_{1}\left(u_{1}\right)\right\} \\
 & =\sup_{u_{1},u_{2}>0}\left\{ u_{2}\cdot\left\{ u_{1}v-K_{1}\left(u_{1}\right)\right\} v-K_{2}\left(u_{2}\right)\right\} \\
 & =\sup_{u_{2}>0}\left\{ u_{2}\cdot\sup_{u_{1}>0}\left\{ u_{1}v-K_{1}\left(u_{1}\right)\right\} -K_{2}\left(u_{2}\right)\right\} .
\end{align*}
Taking the inner supremum simplifies this expression to
\[
K^{*}\left(v\right)=\sup_{u_{2}>0}\left\{ u_{2}\cdot K_{1}^{*}\left(v\right)-K_{2}\left(u_{2}\right)\right\} ,
\]
and taking the remaining supremum allows us to conclude that $K^{*}\left(v\right)=K_{2}^{*}\left(w\right)$
with $w=K_{1}^{*}\left(v\right)$ as claimed, i.e. $K^{*}=K_{2}^{*}\circ K_{1}^{*}$.
\end{proof}
\begin{example}
\label{exa:chaining-two-polynomially-ergodic}If one has $\beta_{i}\left(s\right)=c_{i}s^{-\alpha_{i}}$
for $i\in\{1,2\}$, then $\beta(s)\propto s^{-\alpha_{*}}$ with $\alpha_{*}=\frac{\alpha_{1}\alpha_{2}}{1+\alpha_{1}+\alpha_{2}}$.
To see this, write
\begin{align*}
\beta(s) & :=\inf\left\{ s_{1}\beta_{2}(s_{2})+\beta_{1}(s_{1})\,|s_{1}>0,s_{2}>0,s_{1}s_{2}=s\right\} \\
 & =\inf\left\{ c_{2}s_{1}s_{2}^{-\alpha_{2}}+c_{1}s_{1}^{-\alpha_{1}}\,|s_{1}>0,s_{2}>0,s_{1}s_{2}=s\right\} \\
 & =\inf\left\{ c_{2}s_{1}\left(\frac{s_{1}}{s}\right)^{\alpha_{2}}+c_{1}s_{1}^{-\alpha_{1}}\,|s_{1}>0\right\} \\
 & =\inf\left\{ c_{2}s^{-\alpha_{2}}s_{1}^{1+\alpha_{2}}+c_{1}s_{1}^{-\alpha_{1}}\,|s_{1}>0\right\} .
\end{align*}
Taking derivatives and solving for a stationary point gives $\hat{s}_{1}=\frac{c_{1}\alpha_{1}}{c_{2}\left(1+\alpha_{2}\right)}s^{\frac{\alpha_{2}}{1+\alpha_{1}+\alpha_{2}}}$,
from which point routine algebraic manipulations confirm the conclusion.
\end{example}

Our next main result is Theorem~\ref{thm:practical-comparison-P1-P2},
which provides us with a practical way of establishing (\ref{eq:nested_WPI})
for $T_{i}=P_{i}$. We first establish an intermediate result.
\begin{prop}
\label{prop:Phi-is-2p-norm} Let $P$ be a $\mu$-invariant Markov
kernel, and let $A\in\mathcal{F}\otimes\mathcal{F}$. Let $p\in(1,\infty],q\ge1$
satisfy $p^{-1}+q^{-1}=1$. Then, one can bound for $f\in\ELL(\mu)$,

\[
\int_{A}\mu\left(\dif x\right)P\left(x,\dif y\right)\left(f\left(x\right)-f\left(y\right)\right)^{2}\leqslant\mu\otimes P\big(A\cap\{X\neq Y\}\big)^{1/q}\cdot\Phi_{p}\left(f\right),
\]
with $\Phi_{p}$ given by 
\begin{equation}
\Phi_{p}(f):=\begin{cases}
4\|f\|_{2p}^{2} & p\in(1,\infty)\\
\|f\|_{\osc}^{2} & p=\infty
\end{cases}.\label{eq:=00005CPhi_p}
\end{equation}
Moreover, it holds that for all $f\in\ELL(\mu)$ and $p\in[1,\infty]$,
$\Phi_{p}\left(Pf\right)\leqslant\Phi_{p}\left(f\right)$.
\end{prop}

\begin{proof}
For $p\in(1,\infty)$, we use Hölder's inequality to write,
\begin{align*}
\int_{A}\mu\left(\dif x\right) & P\left(x,\dif y\right)\left(f\left(x\right)-f\left(y\right)\right)^{2}\\
= & \int\mu\left(\dif x\right)P\left(x,\dif y\right)\left\{ \mathbb{I}_{A\cap\{x\neq y\}}(x,y)\cdot\left(f\left(x\right)-f\left(y\right)\right)^{2}\right\} \\
\leqslant & \left(\int\mu\left(\dif x\right)P\left(x,\dif y\right)\mathbb{I}_{A\cap\{x\neq y\}}(x,y)\right)^{1/q}\cdot\left(\int\mu\left(\dif x\right)P\left(x,\dif y\right)\left|f\left(x\right)-f\left(y\right)\right|^{2p}\right)^{1/p}.
\end{align*}
By Jensen's inequality, one can check that $\left|f\left(x\right)-f\left(y\right)\right|^{2p}\leqslant2^{2p-1}\cdot\left(\left|f\left(x\right)\right|^{2p}+\left|f\left(y\right)\right|^{2p}\right)$.
Because $\mu P=\mu$,
\begin{align*}
\int\mu\left(\dif x\right)P\left(x,\dif y\right)\left|f\left(x\right)-f\left(y\right)\right|^{2p} & \leqslant2^{2p-1}\cdot\int\mu\left(\dif x\right)P\left(x,\dif y\right)\cdot\left(\left|f\left(x\right)\right|^{2p}+\left|f\left(y\right)\right|^{2p}\right)\\
 & =2^{2p}\cdot\|f\|_{2p}^{2p}.
\end{align*}
One then concludes that
\begin{align*}
\int_{A}\mu\left({\rm d}x\right)P\left(x,\dif y\right)\left(f\left(x\right)-f\left(y\right)\right)^{2} & \leqslant\mu\otimes P\left(A\right)^{1/q}\cdot\left(2^{2p}\cdot\|f\|_{2p}^{2p}\right)^{1/p}\\
 & =\mu\otimes P\left(A\right)^{1/q}\cdot\Phi\left(f\right),
\end{align*}
as desired. The non-expansivity of $\Phi$ under the action of $P$
can be deduced by writing

\begin{align*}
\|Pf\|_{2p}^{2p} & =\int\mu\left(\dif x\right)\left|Pf\left(x\right)\right|^{2p}\\
 & =\int\mu\left(\dif x\right)\left|\int P\left(x,\dif y\right)f\left(y\right)\right|^{2p}\\
 & \leqslant\int\mu\left(\dif x\right)P\left(x,\dif y\right)\left|f\left(y\right)\right|^{2p}\\
 & =\int\mu\left(\dif y\right)\left|f\left(y\right)\right|^{2p}\\
 & =\|f\|_{2p}^{2p},
\end{align*}
where the inequality uses Jensen's inequality against the probability
measure $P\left(x,\cdot\right)$, and the penultimate equality uses
the $\mu$-invariance of $P$.

When $p=\infty$, we use an analogous argument, noting that $\left(f\left(x\right)-f\left(y\right)\right)^{2}\le\|f\|_{\osc}^{2}$
almost everywhere. 
\end{proof}
\begin{thm}
\label{thm:practical-comparison-P1-P2}Let $P_{1}$ and $P_{2}$ be
two $\mu-$invariant Markov kernels. Assume that for any $(x,A)\in\mathsf{E}\times\mathcal{F}$,
$P_{2}(x,A\setminus\{x\})\geq\int_{A\setminus\{x\}}\varepsilon(x,y)P_{1}(x,{\rm d}y)$
for some $\varepsilon\colon\mathsf{E}^{2}\rightarrow(0,\infty)$.
Then for any $p\in(1,\infty],q\ge1$ such that $p^{-1}+q^{-1}=1$,
any $s>0$, and any $f\in\mathrm{\mathrm{L}}_{0}^{2p}(\mu)\subset\ELL_{0}(\mu)$,
\[
\mathcal{E}(P_{1},f)\leq s\mathcal{E}(P_{2},f)+\frac{1}{2}\cdot\mu\otimes P_{1}\big(A(s)^{\complement}\cap\{X\neq Y\}\big)^{1/q}\Phi_{p}(f),
\]
with $A(s):=\left\{ (x,y)\in\mathsf{E}^{2}\colon s\,\varepsilon(x,y)>1\right\} $
and $\Phi_{p}(f)$ as in (\ref{eq:=00005CPhi_p}), which satisfies
(\ref{eq:phi_condn}).
\end{thm}

\begin{proof}
For any $s>0$ we have
\begin{align*}
\mathcal{E}(P_{1},f) & \leq\frac{1}{2}\int_{A(s)}s\,\varepsilon(x,y)\mu({\rm d}x)P_{1}(x,{\rm d}y)\left[f(y)-f(x)\right]^{2}\\
 & \hspace{4cm}+\frac{1}{2}\int_{A(s)^{\complement}}\mu({\rm d}x)P_{1}(x,{\rm d}y)\left[f(y)-f(x)\right]^{2}\\
 & \leq\frac{s}{2}\int_{A(s)}\mu({\rm d}x)P_{2}(x,{\rm d}y)\left[f(y)-f(x)\right]^{2}+\frac{1}{2}\mu\otimes P_{1}\big(A(s)^{\complement}\cap\{X\neq Y\}\big)^{1/q}\Phi_{p}(f),
\end{align*}
where we have used the assumed inequality between $P_{1}$ and $P_{2}$
and Proposition~\ref{prop:Phi-is-2p-norm}.
\end{proof}
\begin{rem}
\textcolor{black}{Assume for simplicity that for $\mu$-almost all
$x$, $P_{1}(x,\cdot)\equiv P_{2}(x,\cdot)$, i.e. $P_{1}(x,\cdot)$
and $P_{2}(x,\cdot)$ are equivalent measures. This implies $\mu\otimes P_{1}\equiv\mu\otimes P_{2}$
and we may take $\varepsilon(x,y)=\frac{{\rm d}P_{2}(x,\cdot)}{{\rm d}P_{1}(x,\cdot)}(y)=\frac{{\rm d}\mu\otimes P_{2}}{{\rm d}\mu\otimes P_{1}}(x,y)$
to be positive $\mu\otimes P_{1}$-almost everywhere, and we can write
\[
A(s)^{\complement}\cap\{(x,y):x\neq y\}=\left\{ (x,y)\in\mathsf{E}^{2}:\mathbf{1}\{x\neq y\}\frac{{\rm d}\mu\otimes P_{1}}{{\rm d}\mu\otimes P_{2}}(x,y)\geq s\right\} .
\]
Hence, $\lim_{s\to\infty}\mu\otimes P_{1}(A(s)^{\complement}\cap\{X\neq Y\})=0$.
Therefore, Theorem~\ref{thm:practical-comparison-P1-P2} covers many
cases where $P_{2}$ places mass on the same sets as $P_{1}$. In
fact, it covers slightly more general cases in which we only have
$P_{1}(x,A\setminus\{x\})>0\Rightarrow P_{2}(x,A\setminus\{x\})>0$
for $\mu$-almost all $x$ and all $A\in\mathcal{F}$. }
\end{rem}

\begin{rem}
Our result concerned with the IMH algorithm in Subsection~\ref{subsec:Illustration:-Independent-MH}
is a particular case where $P_{{\rm IMH}}(x,A\setminus\{x\})\geq\int_{A\setminus\{x\}}[w^{-1}(x)\wedge w^{-1}(y)]\pi({\rm d}y)$,
which combined with Proposition~\ref{prop:chaining-with-SPI} with
$C_{{\rm P}}=1$ leads to Proposition~\ref{prop:IMH-WPI}.
\end{rem}

\begin{rem}
\textcolor{black}{We note that this approach to identifying weak Poincaré
inequalities can also be generalised to the setting of continuous-time
Markov processes. To this end, consider a continuous-time Markov process
with infinitesimal generator $\mathcal{L}$, and recall the definition
of the carré du champ operator}

\textcolor{black}{
\[
\Gamma\left(f,g\right)\left(x\right):=\frac{1}{2}\left\{ \mathcal{L}\left(fg\right)-\left(\mathcal{L}f\right)\cdot g-f\cdot\left(\mathcal{L}g\right)\right\} .
\]
Note that $\Gamma$ is bilinear and that for all suitable functions
$f$, the function $\Gamma\left(f\right):=\Gamma\left(f,f\right)$
is pointwise nonnegative. }

\textcolor{black}{Suppose now that for two processes with the same
invariant measure $\pi$ and infinitesimal generators given by $\mathcal{L}_{1}$
and $\mathcal{L}_{2}$ respectively, their carré du champ operators
can be ordered pointwise as}

\textcolor{black}{
\[
\Gamma_{1}\left(f\right)\left(x\right)\geqslant w\left(x\right)\cdot\Gamma_{2}\left(f\right)\left(x\right)
\]
for some nonnegative function $w$ (note that the subscripts here
simply index the processes, and have no relation to so-called `iterated
carré du champ' operators). }

\textcolor{black}{Defining $A\left(s\right)=\left\{ x:s\cdot w\left(x\right)\geqslant1\right\} $,
one can then compute that}

\textcolor{black}{
\begin{align*}
\mathcal{E}\left(\mathcal{L}_{2},f\right) & =\int\mu\left(\mathrm{d}x\right)\cdot\Gamma_{2}\left(f\right)\left(x\right)\\
 & \leqslant s\cdot\int_{A\left(s\right)}\mu\left(\mathrm{d}x\right)\cdot w\left(x\right)\cdot\Gamma_{2}\left(f\right)\left(x\right)+\int_{A\left(s\right)^{\complement}}\mu\left(\mathrm{d}x\right)\cdot\Gamma_{2}\left(f\right)\left(x\right)\\
 & \leqslant s\cdot\int\mu\left(\mathrm{d}x\right)\cdot\Gamma_{1}\left(f\right)\left(x\right)+\pi\left(A\left(s\right)^{\complement}\right)\cdot\sup_{x\in\E}\left\{ \Gamma_{2}\left(f\right)\left(x\right)\right\} \\
 & =s\cdot\mathcal{E}\left(\mathcal{L}_{1},f\right)+\beta\left(s\right)\cdot\Phi\left(f\right),
\end{align*}
where we have defined
\begin{align*}
\beta\left(s\right) & :=\mu\big(A\left(s\right)^{\complement}\big)\\
\Phi\left(f\right) & :=\sup_{x\in\E}\left\{ \Gamma_{2}\left(f\right)\left(x\right)\right\} .
\end{align*}
Note that in many applications, $\Gamma\left(f\right)$ has the character
of a squared gradient, and hence $\Phi\left(f\right)$ will behave
much like a squared Lipschitz constant for the function $f$.}

\textcolor{black}{Comparisons of this form have been used implicitly
in the study of so-called weighted and converse weighted Poincaré
inequalities \cite{Bobkov09,Cattiaux10-1}, which are known to imply
weak Poincaré inequalities. Such comparison inequalities can then
be applied to compare the convergence of continuous-time processes
in much the same fashion as in this work.}
\end{rem}

The following generalizes the criterion $P(x,\{x\})\geq\varepsilon$
for some $\varepsilon>0$ and all $x\in\mathsf{E}$ often used to
establish the existence of a left spectral gap for reversible Markov
chains.
\begin{thm}
\label{thm:weakly-lazy-chain}Assume that $P$ is $\mu-$invariant
and that for any $(x,A)\in\mathsf{E}\times\mathcal{F}$ it satisfies
$P(x,A)\geq\varepsilon(x)\int_{A}\delta_{x}({\rm d}y)$ for some $\varepsilon\colon\mathsf{E}\rightarrow[0,1]$.
Then 
\begin{enumerate}
\item for any $(x,A)\in\mathsf{E}\times\mathcal{F}$, $P^{2}(x,A)\geq\varepsilon(x)P(x,A)$,
\item for any $p\in(1,\infty]$, $1/q=1-1/p$, any $f\in\mathrm{\mathrm{L}}_{0}^{2p}(\mu)\subset\ELL_{0}(\mu)$
and $s>0$
\[
\mathcal{E}(P,f)\leq s\mathcal{E}(P^{2},f)+\frac{1}{2}\mu\big(\varepsilon(X)^{-1}\geq s\big)^{1/q}\Phi_{p}(f).
\]
\end{enumerate}
\end{thm}

\begin{cor}
Proposition~\ref{prop:chaining-with-SPI} or Theorem~\ref{thm:chaining-with-WPI}
can be applied with $T_{1}=P_{1}=P$ and $T_{2}=P_{2}=P^{2}$. This
can be applied to the Metropolis--Hastings (MH) algorithm (see (\ref{eq:MH-kernel}))
as soon as $\mu(\rho(X)>0)=1$ and also means that weakly lazy chains
can be defined as $\varepsilon(x){\rm Id}+\big(1-\varepsilon(x)\big)\check{P}$
where $\check{P}$ is a MH using proposal $P$.
\end{cor}

\begin{proof}[Proof of Theorem \ref{thm:weakly-lazy-chain}]
For $(x,A)\in\mathsf{E}\times\mathcal{F}$,
\[
P^{2}(x,A)=\int P(x,{\rm d}y)P(y,A)\geq\varepsilon(x)P(x,A)
\]
 and we apply Theorem~\ref{thm:practical-comparison-P1-P2}. Now
$\mu\otimes P\big(\{\varepsilon(X)^{-1}\geq s\}\cap\{X\neq Y\}\big)\leq\mu(\varepsilon(X)^{-1}\geq s)$
and we conclude.
\end{proof}
\begin{rem}
It is natural to consider the scenario where $T_{1}$ and $T_{2}$
admit different invariant distributions, $\mu_{1}$ and $\mu_{2}$
respectively. With straightforward notation for $\|\cdot\|_{2}$ under
$\mu_{1}$ and $\mu_{2}$, a condition of the type
\begin{equation}
\|f-\mu_{2}(f)\|_{\mu_{2}}^{2}\leq s\|f-\mu_{1}(f)\|_{\mu_{1}}^{2}+\beta(s)\Phi(f),\label{eq:mu1-mu2-poincare}
\end{equation}
is natural. Using polarisation one can show that, with $A(s)=\{(x,y)\in\mathsf{E}^{2}\colon{\rm d}(\mu_{1}\otimes\mu_{1})/{\rm d}(\mu_{2}\otimes\mu_{2})(x,y)\geq1/s\}$,
\begin{multline*}
\frac{1}{2}\int\big(f(y)-f(x)\big)^{2}\dif\left(\mu_{2}\otimes\mu_{2}\right)(x,y)\leq\frac{s}{2}\int_{A(s)}\big(f(y)-f(x)\big)^{2}\dif\left(\mu_{1}\otimes\mu_{1}\right)(x,y)\\
+\|f\|_{{\rm osc}}^{2}\frac{\mu_{2}\otimes\mu_{2}}{2}\left(\frac{{\rm d}(\mu_{2}\otimes\mu_{2})}{{\rm d}(\mu_{1}\otimes\mu_{1})}(X,Y)>s\right),
\end{multline*}
implying (\ref{eq:mu1-mu2-poincare}). We note however that in the
examples we have considered, we have found the last term to lead to
poor convergence rates.
\end{rem}

Our final result in this section concerns the situation when one has
a sequence of \textcolor{black}{weak} Poincaré inequalities given
by functions $\left\{ \beta_{2,\iota}\right\} _{\iota>0}$, which
converge pointwise to an appropriate function $\beta_{1}$. We give
conditions under which the corresponding convergence rates $F_{2,\iota}^{-1}$
will also converge to the corresponding $F_{1}^{-1}$.
\begin{prop}
\label{prop:seq_beta_iota}Let $P_{1}$ and $(P_{2,\iota})_{\iota>0}$
be $\mu-$invariant Markov kernels. Assume $P_{1}$ satisfies a \textcolor{black}{weak}
Poincaré inequality with function $\beta_{1}$ and that for any $\iota>0$
\[
\|f\|_{2}^{2}\leq s\mathcal{E}(P_{2,\iota}^{*}P_{2,\iota},f)+\beta_{2,\iota}(s)\Phi(f),\quad\forall s>0,
\]
where each $\beta_{2,\iota}$ satisfies the conditions in Definition~\ref{def:Super-Poincar=0000E9-inequality}.
Let $F_{1},F_{2,\iota}:(0,a]\to[0,\infty)$ for each $\iota>0$ be
as defined in Section~\ref{sec:Weak-Poincar=0000E9-inequalities-overview}. 

Assume that for any $\iota>0$, $\beta_{2,\iota}\geq\beta_{1}$ pointwise
and for any $s>0$, $\lim_{\iota\rightarrow0}\beta_{2,\iota}(s)=\beta_{1}(s).$
Then for any $\iota>0$ and $n\in\mathbb{N}$, $F_{2,\iota}^{-1}(n)\geq F_{1}^{-1}(n)$
and
\[
\lim_{\iota\rightarrow0}\sup_{n\geq0}\big\{ F_{2,\iota}^{-1}(n)-F_{1}^{-1}(n)\big\}=0.
\]
\end{prop}

\begin{proof}
Let $v>0$ and $(u_{n})$ be such that $K_{1}^{*}(v)=\lim_{n\rightarrow\infty}u_{n}[v-\beta_{1}(1/u_{n})]$.
Then for any $\iota>0$ and any $n\geq1$, $K_{2,\iota}^{*}(v)\geq u_{n}[v-\beta_{2,\iota}(1/u_{n})],$
and therefore 
\[
\liminf_{\iota\rightarrow0}K_{2,\iota}^{*}(v)\geq\lim_{\iota\rightarrow0}u_{n}[v-\beta_{2,\iota}(1/u_{n})]=u_{n}[v-\beta_{1}(1/u_{n})].
\]
Consequently,
\[
\liminf_{\iota\rightarrow0}K_{2,\iota}^{*}(v)\geq\lim_{n\rightarrow\infty}u_{n}[v-\beta_{1}(1/u_{n})]=K_{1}^{*}(v).
\]
Since for any $\iota>0$, $\beta_{2,\iota}\geq\beta_{1}$ implies
$K_{2,\iota}^{*}\leq K_{1}^{*}$, we have $\limsup_{\iota\rightarrow0}K_{2,\iota}^{*}(v)\leq K_{1}^{*}(v)$.
We therefore conclude that $\lim_{\iota\rightarrow0}K_{2,\iota}^{*}(v)=K_{1}^{*}(v)$.
Now let $0<x\leq a$, and choose $\epsilon>0$ such that $K_{1}(x)-\epsilon>0$.
Then there exists $\iota_{0}>0$ such that for any $0<\iota\leq\iota_{0}$,
$K_{2,\iota}^{*}(x)\geq K_{1}(x)-\epsilon>0$, and since $v\mapsto K_{2,\iota}^{*}(v)$
is increasing, we deduce $0<\sup_{v\in[x,a]}\big(K_{2,\iota}^{*}(v)\big)^{-1}\leq\big(K_{2,\iota}^{*}(x)\big)^{-1}\leq\big(K_{1}^{*}(x)-\epsilon\big)^{-1}<\infty$.
We can therefore apply the bounded convergence theorem and conclude
\[
\lim_{\iota\rightarrow0}\int_{x}^{a}\frac{{\rm d}v}{K_{2,\iota}^{*}(v)}=F_{1}(x).
\]
For any $\iota>0$, $F_{1},F_{2,\iota}\colon(0,a]\rightarrow[0,\infty)$
are decreasing and continuous and so are the inverse functions $F_{1}^{-1},F_{2,\iota}^{-1}\colon[0,\infty)\rightarrow(0,a]$,
and consequently for any $x\in[0,\infty)$, $\lim_{\iota\rightarrow0}F_{2,\iota}^{-1}(x)=F_{1}^{-1}(x)$
(note that $F_{2,\iota}^{-1}(0)=F_{1}^{-1}(0)=a$). Since $K_{2,\iota}^{*}\leq K_{1}^{*}$,
we immediately have the ordering $F_{2,\iota}^{-1}(x)\geq F_{1}^{-1}(x)$
for any $x\in[0,\infty)$.

Now let $\epsilon>0$, then there exists $n_{0}\in\mathbb{N}$ such
that for any $n\geq n_{0}$, $F_{1}^{-1}(n)\leq\epsilon/2$. From
the convergence above, there exists $\iota_{0}>0$ such that for any
$0<\iota\leq\iota_{0}$,
\[
0\leq F_{2,\iota}^{-1}(n_{0})-F_{1}^{-1}(n_{0})\leq\epsilon/2,
\]
and since $n\mapsto F_{2,\iota}^{-1}(n)$ is decreasing, for any $n\ge n_{0}$,
\[
F_{2,\iota}^{-1}(n)-F_{1}^{-1}(n)\leq F_{2,\iota}^{-1}(n)\leq\epsilon.
\]
Now, there exists $\iota'_{0}>0$ such that for any $0<\iota<\iota'_{0}$,
\[
\max_{0\leq n<n_{0}}\left\{ F_{2,\iota}^{-1}(n)-F_{1}^{-1}(n)\right\} \leq\epsilon.
\]
Therefore,
\[
\lim_{\iota\rightarrow0}\sup_{n\geq0}\left\{ F_{2,\iota}^{-1}(n)-F_{1}^{-1}(n)\right\} =0.
\]
\end{proof}
\begin{example}
\label{exa:cvg-super-to-strong}Let $\beta_{1}(s)=a\mathbb{I}\{s\leq C_{{\rm P}}^{-1}\}$,
which corresponds to $P_{1}$ satisfying a strong Poincaré inequality
as in Remark~\ref{rmk:strong-PI-super-beta}. Let $P_{2,\iota}$
satisfy a \textcolor{black}{weak} Poincaré inequality with $\beta_{2,\iota}(s)=a\wedge\{\beta'_{\iota}(C_{{\rm P}}s)/C_{{\rm P}}\}$
(e.g. by Proposition~\ref{prop:chaining-with-SPI}) where $\beta'_{\iota}(s)\geq\mathbb{I}\{s\leq1\}$
so that $\beta_{2,\iota}\geq\beta_{1}$. If $\lim_{\iota\to0}\beta'_{\iota}(s)=\mathbb{I}\{s\leq1\}$
for all $s>0$, then Proposition~\ref{prop:seq_beta_iota} can be
applied and we recover exponential convergence as $\iota\rightarrow0$,
and for any $\epsilon>0$ the existence of $\iota>0$, such that $F_{2,\iota}^{-1}(n)-F_{1}^{-1}(n)<\epsilon$
for all $n$. In other words, one may obtain convergence for $P_{2}$
arbitrarily close to that given by $\beta_{1}$ by taking $\iota$
sufficiently small.
\end{example}

\section{\label{sec:Application-to-pseudo-marginal}Application to pseudo-marginal
methods}

We now present our main application. Fix a probability distribution
$\pi$ on a measure space $\mathsf{X}$, with a density function also
denoted $\pi$. Pseudo-marginal algorithms \cite{Andrieu09} extend
the scope of the Metropolis--Hastings algorithms to the scenario
where the density $\pi$ is intractable, but for any $x\in\mathsf{X}$,
nonnegative estimators $\hat{\pi}(x)$ such that $\mathbb{E}[\hat{\pi}(x)]=C\pi(x)$
for some constant $C>0$ are available. This can be conveniently formulated
as $\tilde{\pi}({\rm d}x,{\rm d}w)=\pi({\rm d}x)Q_{x}({\rm d}w)w=\pi({\rm d}x)\tilde{\pi}_{x}({\rm d}w)$
with $\int_{\mathbb{R}_{+}}wQ_{x}({\rm d}w)=1$ on an extended space
$\E:=\mathsf{X\times\mathbb{R}_{+}}$. We will refer to these auxiliary
$w$ random variables as \emph{weights} or \emph{perturbations}.

\subsection{A weak Poincaré inequality for pseudo-marginal chains }

A question of interest is to characterise the degradation in performance,
compared to the \emph{marginal algorithm}, which uses the exact density
$\pi$. More specifically, for $\{q(x,\cdot),x\in\mathsf{X}\}$ a
family of proposal distributions, the marginal algorithm is described
by the kernel
\begin{align}
P(x,{\rm d}y) & =[1\wedge\mathfrak{r}(x,y)]\,q(x,{\rm d}y)+\delta_{x}({\rm d}y)\rho(x),\nonumber \\
 & \quad\text{where}\quad\mathfrak{r}(x,y):=\frac{\pi(y)q(y,{\rm d}x)}{\pi(x)q(x,{\rm d}y)},\label{eq:MH-kernel}
\end{align}
and $\rho$ is the rejection probability given by $\rho(x):=1-\int[1\wedge\mathfrak{r}(x,y)]\,q(x,\dif y)$
for each $x\in\mathsf{X}$. For brevity we will also define the acceptance
probability as $a(x,y):=\left[1\wedge\mathfrak{r}(x,y)\right]$.

The pseudo-marginal Metropolis--Hastings kernel is given by
\[
\tilde{P}(x,w;{\rm d}y,{\rm d}u)=\left[1\wedge\left\{ \mathfrak{r}(x,y)\frac{u}{w}\right\} \right]q(x,{\rm d}y)Q_{y}({\rm d}u)+\delta_{x,w}({\rm d}y,{\rm d}u)\tilde{\rho}(x,w),
\]
where the (joint) rejection probability $\tilde{\rho}(x,w)$ is analogously
defined. It is known in this context that perturbing the acceptance
ratio of the marginal algorithm leads to a loss in performance, in
particular in terms of convergence rates to equilibrium. More specifically,
if $P$ is geometrically ergodic, then $\tilde{P}$ is geometrically
ergodic if the perturbations are bounded uniformly in $x$, and cannot
be geometrically ergodic if the perturbations are unbounded on a set
of positive $\pi$-probability, which addressed in \cite{Andrieu15}
in specific scenarios using Foster--Lyapunov and minorisation conditions
by linking the existence of moments of the perturbations to the subgeometric
rate of convergence of the algorithm. When the perturbations are bounded
for each $x$ but not bounded uniformly in $x$, the situation is
more complicated: if ``local proposals'' are used then \cite{Lee14}
proves that $\tilde{P}$ cannot be geometric under fairly weak assumptions
in statistical applications whereas if global proposals are used $\tilde{P}$
may still be geometric (consider, for instance, the setting of \cite[Remark~5]{deli-2018}).
We show here that convergence results can be made completely general
using weak Poincaré inequalities, with much simpler and considerably
more transparent arguments.

We will be assuming throughout this section that the pseudo-marginal
kernel $\tilde{P}$ is positive, in order to utilize our results from
Section~\ref{subsec:Simplified-weak-Poincar=0000E9}. We note that
this positivity assumption is not restrictive; as established in \cite[Proposition 16]{Andrieu15},
$\tilde{P}$ will be positive if the marginal chain $P$ is an Independent
MH sampler, or a random walk Metropolis kernel with multivariate Gaussian
or student-$t$ increments.

The following comparison theorem plays a central role.
\begin{thm}
\label{thm:dirichlet-comparisons}Let $\bar{P}$ be the embedding
of $P$ in the joint space $\E=\mathsf{X}\times\mathbb{R}_{+}$,
\[
\bar{P}(x,w;{\rm d}y,{\rm d}u):=a(x,y)q(x,{\rm d}y)\tilde{\pi}_{y}({\rm d}u)+\delta_{x,w}({\rm d}y,{\rm d}u)\rho(x).
\]
Then for any $p\in(1,\infty],q\ge1$ such that $p^{-1}+q^{-1}=1$,
any $s>0$, and any $f\in\mathrm{L}^{2p}(\tilde{\pi})\subset\ELL_{0}(\tilde{\pi})$,
\[
\mathcal{E}(\bar{P},f)\leq s\mathcal{E}(\tilde{P},f)+\frac{1}{2}\cdot\Phi_{p}(f)\left(2\int_{\mathsf{X}}\tilde{\pi}_{x}(w\geq s)\,\pi({\rm d}x)\right)^{1/q},
\]
with $\Phi_{p}(f)$ given in (\ref{eq:=00005CPhi_p}).
\end{thm}

\begin{proof}
We apply Theorem~\ref{thm:practical-comparison-P1-P2}. Let $\varepsilon(w,u):=w^{-1}\wedge u^{-1}$,
then for any $(x,w)\in\mathsf{E}$ and $B\in\mathcal{F}$,
\begin{align*}
\int_{B\setminus\{x,w\}}\varepsilon(w,u)\bar{P}(x,w;{\rm d}y,{\rm d}u)\\
=\int_{B} & q(x,\dif y)\tilde{\pi}_{y}({\rm d}u)a(x,y)(w^{-1}\wedge u^{-1})\\
=\int_{B} & q(x,{\rm d}y)Q_{y}({\rm d}u)u\,a(x,y)(w^{-1}\wedge u^{-1})\\
=\int_{B} & q(x,{\rm d}y)Q_{y}({\rm d}u)a(x,y)(1\wedge\frac{u}{w})\\
\leq\int_{B} & q(x,{\rm d}y)Q_{y}({\rm d}u)[1\wedge\bigl(\mathfrak{r}(x,y)\frac{u}{w}\bigr)],\\
= & \tilde{P}(x,w;B\setminus\{x,w\}).
\end{align*}
where we have used that $1\wedge(ab)\geq(1\wedge a)(1\wedge b)$ for
$a,b\geq0$. Now for $s>0$ let
\begin{align*}
A(s) & :=\left\{ (w,u)\in\R_{+}^{2}\colon w^{-1}\wedge u^{-1}>1/s\right\} ,\\
\bar{A}(s) & :=\left\{ (x,w,y,u)\in\mathsf{E}\times\mathsf{E}\colon w^{-1}\wedge u^{-1}>1/s\right\} .
\end{align*}
Then,
\begin{align*}
\mu\otimes\bar{P}\big(\bar{A}^{\complement}(s)\cap\{(X,W)\neq(Y,U)\}\big)\leq & \int_{\bar{A}(s)^{\complement}}a(x,y)\pi({\rm d}x)q(x,{\rm d}y)\tilde{\pi}_{x}({\rm d}w)\tilde{\pi}_{y}({\rm d}u)\\
= & \int_{\mathsf{X}^{2}}\left[a(x,y)\int_{A(s)^{\complement}}\tilde{\pi}_{x}(\dif w)\tilde{\pi}_{y}(\dif u)\right]\pi({\rm d}x)q(x,{\rm d}y),
\end{align*}
and
\begin{align*}
\int_{A(s)^{\complement}}\tilde{\pi}_{x}({\rm d}w)\tilde{\pi}_{y}({\rm d}u) & =1-\tilde{\pi}_{x}(w\leq s)\tilde{\pi}_{y}(u\leq s)\\
 & =1-[1-\tilde{\pi}_{x}(w>s)][1-\tilde{\pi}_{y}(u>s)]\\
 & \leq\tilde{\pi}_{x}(w\geq s)+\tilde{\pi}_{y}(u\geq s).
\end{align*}
Therefore,
\[
\mu\otimes\bar{P}\big(\bar{A}(s)^{\complement}\cap\{(X,W)\neq(Y,U)\}\big)\leq2\int\tilde{\pi}_{x}(w\geq s)\pi({\rm d}x).
\]
We conclude.
\end{proof}

We are now in a position to apply Proposition\ \ref{prop:chaining-with-SPI}
or Theorem\ \ref{thm:chaining-with-WPI}. We will see that the tail
behaviour of the perturbations governs the rate at which our bound
on $\|\tilde{P}^{n}f\|_{2}$ vanishes as $n\rightarrow\infty$. For
simplicity, for the remainder of this section we focus on the case
where $\Phi(f)=\|f\|_{{\rm osc}}^{2}$.
\begin{cor}
\label{cor:pseudo-mar_bound}When $\bar{P}$ satisfies a strong Poincaré
inequality with constant $C_{\mathrm{P}}$ as in Remark~\ref{rem:strong_PI},
Proposition~\ref{prop:chaining-with-SPI} establishes that $\tilde{P}$
satisfies Definition~\ref{def:WPI_rev} with $\beta(s)=\beta'(C_{\mathrm{P}}s)/C_{P}$
where $\beta'(s)=\int\tilde{\pi}_{x}(w\geq s)\pi({\rm d}x)$ and $\Phi(f)=\|f\|_{{\rm osc}}^{2}$.
Consequently, Theorem~\ref{thm:rev_pos_conv} applies to $\tilde{P}$
with a rate determined by $\beta(s)$. 

Furthermore, using Markov's inequality, the existence of moments of
$W$ under $\tilde{\pi}_{x}$ of order $k\in\mathbb{N_{*}}$ implies
\[
\beta'(s)\leq s^{-k}\int_{\mathsf{X}}\mathbb{E}_{\tilde{\pi}_{x}}\left[|W|^{k}\right]\pi({\rm d}x).
\]
Provided the integral is finite, this leads to a polynomial rate of
convergence $O(n^{-k})$ by Lemma~\ref{lem:rate-beta-decays-polynomial}.

Similarly, if $\bar{P}$ satisfies a weak Poincaré inequality, one
can apply Theorem~\ref{thm:chaining-with-WPI} and deduce the new
rate of convergence as in Example~\ref{exa:chaining-two-polynomially-ergodic}.
\end{cor}

\begin{rem}
\label{rem:pm-bounded-strong-pi}Notice that when the perturbations
are uniformly bounded, i.e. there exists $\bar{w}$ such that for
all $x\in\mathsf{X}$, $\tilde{\pi}_{x}(w\geq\bar{w})=0$, and $\bar{P}$
satisfies a strong Poincaré inequality, then $C_{{\rm P}}\|f\|^{2}\leq\mathcal{E}(\bar{P},f)\leq\bar{w}\mathcal{E}(\tilde{P},f)$
and we recover the known results of \cite{Andrieu09,Andrieu15}. 

\textcolor{black}{Examples of chains for which $\bar{P}$ satisfies
a strong Poincaré inequality are numerous; the IMH and random walk
Metropolis algorithms often possess a spectral gap; see \cite{Andrieu15}
where these examples are considered in the context of pseudo-marginal
algorithms.}
\end{rem}

We provide a general result demonstrating that under very weak conditions
pseudo-marginal convergence can be made arbitrarily close to marginal
convergence, strengthening the result of \cite[Section~4]{Andrieu09}.
\begin{rem}
\label{rem:pseudo-cv-marginal-general}Assume that there is a parameter
$\iota>0$ controlling the quality of the perturbations $W\sim\tilde{\pi}_{x,\iota}$
such that for each $x\in\mathsf{X}$, $W$ converges in probability
to $1$ as $\iota\to0$:
\[
\lim_{\iota\to0}\tilde{\pi}_{x,\iota}(w\geq s)=\mathbb{I}\{s\leq1\},\qquad x\in\mathsf{X},s>0.
\]
Let
\[
\beta_{\iota}(s):=\int_{\mathsf{X}}\tilde{\pi}_{x,\iota}(w\geq s)\,\pi({\rm d}x),
\]
then the bounded convergence theorem implies that
\[
\lim_{\iota\to0}\beta_{\iota}(s)=\mathbb{I}\{s\leq1\},\qquad s>0.
\]
Assume now that $\bar{P}$ satisfies a weak Poincaré inequality with
function $\bar{\beta}$. Similar to Example~\ref{exa:cvg-super-to-strong},
one can compare the convergence bounds for $\tilde{P}_{\iota}$ and
$\bar{P}$ via their respective functions $\tilde{\beta}_{\iota}$
and $\bar{\beta}$. Indeed, Theorem~\ref{thm:dirichlet-comparisons}
and Theorem~\ref{thm:chaining-with-WPI} imply that $\tilde{P}_{\iota}$
satisfies a \textcolor{red}{weak} Poincaré inequality with function
\[
\tilde{\beta}_{\iota,\epsilon}(s)=\frac{s}{1+\epsilon}\beta_{\iota}(1+\epsilon)+\bar{\beta}\left(\frac{s}{1+\epsilon}\right),
\]
where $\epsilon>0$ is arbitrary. Note that for $s>0$, $\tilde{\beta}_{\iota,\epsilon}(s)\geq\bar{\beta}\big(s/(1+\epsilon)\big)\geq\bar{\beta}(s)$
and since $\lim_{\iota\rightarrow0}\tilde{\beta}_{\iota}(s)=\bar{\beta}\left(s/(1+\epsilon)\right)$
we can apply Proposition~\ref{prop:seq_beta_iota} to obtain convergence
bounds arbitrarily close to those of $\bar{P}$ with rate function
$\bar{\beta}$.
\end{rem}

\subsection{The effect of averaging\label{subsec:The-effect-of}}

A natural idea to reduce the variability of pseudo-marginal chains
is to average several estimators $\hat{\pi}$ of the target density
at each iteration. As pointed out in \cite{Andrieu15}, this is unlikely
to affect asymptotic rates of convergence. Furthermore, it was established
in \cite{bornn2014one,sherlock-2017} that when considering asymptotic
variance, it is preferable to combine the output of $N$ independent
chains each using $1$ estimator, rather than running $1$ chain averaging
$N$ estimators at each iteration. The following, motivated by the
application of Markov's inequality, adds nuance to these conclusions
by showing how bias can be reduced by averaging, particularly in situations
where higher order moments of the perturbations are large.
\begin{lem}
\label{lem:averaging}Let $\{W_{i}\}$ be i.i.d., of expectation $1$
and such that for a given $p\in\mathbb{N}$ with $p\ge2$, $\mathbb{E}\left(|W_{1}|^{p}\right)<\infty$.
Then there are some constants $\left\{ C_{p,k}\right\} $, such that
for any $N\in\mathbb{N}$,
\begin{equation}
\mathbb{E}\left[\left|\frac{1}{N}\sum_{i=1}^{N}W_{i}\right|^{p}\right]\le1+\sum_{k=2}^{p}N^{-k/2}C_{p,k}\Ebb\left[|W_{1}-1|^{k}\right].\label{eq:averaging_bd}
\end{equation}
For large $N$, this bound is $1+O(N^{-1})$.
\end{lem}

As an illustration, we focus here on the scenario where the marginal
chain satisfies a strong Poincaré inequality (Remark~\ref{rem:strong_PI})
and the moments are uniformly bounded in $x\in\mathsf{X}$. Let $\mathcal{W}_{N}:=N^{-1}\sum_{i=1}^{N}W_{i}$,
then Markov's inequality implies that for the pseudo-marginal algorithm
which averages $N$ estimators,
\begin{align*}
\beta'_{N}(s) & \leq\left[\sup_{x\in\mathsf{X}}\mathbb{E}_{\tilde{\pi}_{x}}\big(\mathcal{W}_{N}^{p}\big)\right]s^{-p},
\end{align*}
and while the rate of convergence in $s$ is independent of $N$,
the multiplicative constant in square brackets does depend on $N$.
Indeed, by Lemma~\ref{lem:averaging}, averaging by choosing $N>1$
can reduce its magnitude and reduce our convergence upper bounds in
Theorem~\ref{thm:rev_pos_conv}, thanks to Lemma~\ref{lem:rescaling-beta-rescaling-invF}.
The bound obtained in (\ref{eq:averaging_bd}) suggests that while
the asymptotic rate of decay for large $N$ is governed by the term
$\Ebb\left[\left|W_{1}-1\right|^{2}\right]N^{-1}$, inversely proportional
to the increased computational cost at each MCMC iteration, higher
order moments may play an important role for small to moderate values
of $N$. 

This is expected for heavy-tailed distributions: for example, consider
an inverse gamma distribution of expectation $1$ and shape parameter
$\mathsf{s}>1$. Its raw (polynomial) moments grow very rapidly as
$\mathbb{E}\left(W_{1}^{k}\right)=(\mathsf{s}-1)^{k}/\prod_{i=1}^{k}(\mathsf{s}-i)$
for $k\in\mathbb{N}$, $k<\mathsf{s}$ and $\mathsf{s}$ large, and
for small and moderate values of $N$, summands other than $k=2$
in (\ref{eq:averaging_bd}) will be most prominent.

\subsection{\label{subsec:ABC-example}ABC example}

We consider an Approximate Bayesian Computation (ABC) setting, using
notation inspired by \cite{Lee14}. We assume we have a true posterior
density $\pi_{0}(x)\propto\nu(x)\ell_{y}(x)$ on a space $\mathsf{X}\subset\mathbb{R}^{d_{x}}$,
where $\nu(\cdot)$ represents the prior and $x\mapsto\ell_{y}(x)$
is an intractable likelihood corresponding to a probability density
$f_{x}(y)=\ell_{y}(x)$ for some fixed observations $y\in\mathsf{Y}\subset\R^{\mathrm{d}}$.
It is known that ABC Markov chains of the type considered here cannot
be geometrically ergodic under fairly weak conditions when a ``local
proposal'' is used \cite[Theorem 2]{Lee14}.

Fix an $\epsilon>0$ and $x\in\mathcal{\mathsf{X}}$, and for $j=1,\dots,N$,
let $z_{j}\overset{{\rm iid}}{\sim}f_{x}(\cdot)$ be auxiliary random
variables and define the random variables $W_{j}$, where $\left|\cdot\right|$
denotes the Euclidean norm, 
\[
W_{j}=\begin{cases}
1/\ell_{{\rm ABC}}(x) & \text{if }\left|z_{j}-y\right|<\epsilon,\\
0 & \text{else,}
\end{cases}
\]
with $\ell_{{\rm ABC}}(x):=\mathbb{P}_{x}(\left|z_{1}-y\right|<\epsilon)$.
In an ABC setup, the intractable $\pi_{0}$ is replaced with the ABC
posterior $\pi(x)\propto\nu(x)\ell_{{\rm ABC}}(x)$, which is typically
also intractable and itself approximated using a pseudo-marginal approach:
for fixed $N\in\mathbb{N}$, define
\begin{align*}
\tilde{\pi}(x,z_{1},\ldots,z_{N}) & \propto\nu(x)\ell_{{\rm ABC}}(x)\left[\prod_{j=1}^{N}f_{x}(z_{j})\right]\cdot\frac{1}{N}\sum_{j=1}^{N}W_{j}.
\end{align*}
It is easily seen that $\mathcal{W}_{N}:=\frac{1}{N}\sum_{j=1}^{N}W_{j}$
has expectation 1 under $\big[\prod_{j=1}^{N}f_{x}(z_{j})\big]{\rm d}z_{1}\times\cdots\times{\rm d}z_{N}$
for a fixed $x\in\mathsf{X}$. In our previous notation, $Q_{x}(\dif w)$
is then the law of $\mathcal{W}_{N}$ when the $(z_{1},\ldots,z_{N})$
are drawn from $\big[\prod_{j=1}^{N}f_{x}(z_{j})\big]{\rm d}z_{1}\times\cdots\times{\rm d}z_{N}$,
and $\tilde{\pi}_{x}(\dif w)=wQ_{x}(\dif w)$. Given $x\in\mathsf{X}$,
it is clear that under $Q_{x}$, we have that $\ell_{{\rm ABC}}(x)\sum_{j=1}^{N}W_{j}\sim\mathrm{Bin}\big(N,\ell_{{\rm ABC}}(x)\big)$.

Thus from our previous result, Corollary~\ref{cor:pseudo-mar_bound},
in order to bound the rate of convergence of the resulting pseudo-marginal
algorithm, we need to bound for $s>0$,
\[
\int_{\mathsf{X}}\pi(\dif x)\tilde{\pi}_{x}(\mathcal{W}_{N}\ge s).
\]
So given $x\in\mathcal{\mathsf{X}}$, $s>0$, we first consider $\tilde{\pi}_{x}(\mathcal{W}_{N}\ge s)$.
Using Markov's inequality, for any $p\in\mathbb{N}$, we can bound
\[
\tilde{\pi}_{x}(\mathcal{W}_{N}\ge s)\le\frac{\tilde{\pi}_{x}\left[\mathcal{W}_{N}^{p}\right]}{s^{p}}=\frac{Q_{x}\left[\mathcal{W}_{N}^{p+1}\right]}{s^{p}}.
\]
This seems to suggest that if the marginal algorithm is geometrically
ergodic, then its ABC approximation converges at any polynomial rate.
The following result tells us that this may not be the case.
\begin{prop}
\label{prop:ABC-example}For a given $p\in\mathbb{N}$, suppose that
$\int_{\mathsf{X}}\nu(x)\ell_{{\rm ABC}}^{-(p-1)}(x)\dif x<\infty$.
Then there is $C_{N,p}>0$ such that for all $s>0$,
\[
\int_{\mathcal{\mathsf{X}}}\pi(\dif x)\tilde{\pi}_{x}(\mathcal{W}_{N}\ge s)\le C_{N,p}s^{-p},
\]
and as $N\to\infty$, $C_{N,p}=1+O(1/N)$. In particular, we may always
take $p=1$. The resulting convergence rate for the pseudo-marginal
chain is then also $O(n^{-p})$ as in Lemma~\ref{lem:rate-beta-decays-polynomial}.
\end{prop}

\subsection{Products of averages}

The results in Sections~\ref{subsec:The-effect-of}--\ref{subsec:ABC-example}
suggest that $N$ may not need to be taken too large in the case of
simple averaging. We consider here a scenario where the perturbation
is instead a product of $T$ independent averages, which gives different
conclusions, and can be seen as a simple version of the perturbation
involved in a particle marginal Metropolis--Hastings (PMMH) algorithm
\cite{Andrieu10}, a special case of a pseudo-marginal algorithm.
Such scenarios can arise in random effects and latent variable models.
For example, \cite[Section 4.1]{tran-2016} uses a random effects
model from \cite[Section 6.1]{donohue-2011} to analyze the data from
\cite{greenberg-1990}, while \cite[Section 4.2]{Lee-2019} considers
an ABC example with i.i.d. data and \cite[Section 4.3]{Lee-2019}
considers a single-cell gene expression model proposed by \cite{Peccoud-95}
and employed, e.g., by \cite{Tiberi-2018}. 

The following bound can be used in Corollary~\ref{cor:pseudo-mar_bound},
and indicates that it is sufficient to take $N$ proportional to $T$
to obtain $T$-independent bounds on the relevant tail probabilities
as long as $\pi$ is sufficiently concentrated. 

\begin{prop}
\label{prop:prod-avg-bound}Assume $W\sim Q_{x}$ can be written as
$W=\prod_{t=1}^{T}W_{t}$, where each $W_{t}$ is independent and
nonnegative, and for each $t\in\{1,\ldots,T\}$, 
\[
W_{t}=\frac{1}{N}\sum_{i=1}^{N}W_{t,i},
\]
is an average of nonnegative, identically distributed random variables
with mean $1$. Assume that for some $p\in\mathbb{N}$ with $p\ge2$,
and any $x\in\mathsf{X}$, 
\[
\max_{t\in\{1,\ldots,T\}}\mathbb{E}\left[W_{t,1}^{p}\right]<\infty.
\]
Then there exists a function $M_{p}:\mathsf{X}\to\mathbb{R}_{+}$
such that if we take
\[
N\geq\alpha T+\frac{1}{2}+\sqrt{\alpha T},
\]
for some $\alpha>0$, then
\[
\int\pi({\rm d}x)\tilde{\pi}_{x}(W\geq s)\leq s^{-p+1}\int\pi({\rm d}x)\exp\left(\frac{M_{p}(x)}{\alpha}\right),
\]
where the right-hand side may be finite or infinite depending on $\pi$.
\end{prop}

In particular, we can see that if the function $M_{p}$ grows quickly
in the tails of $\pi$, then the bound is finite only if $\pi$ has
sufficiently light tails.
\begin{example}
Assume $M_{p}(x)=bx^{k}$ and $\pi({\rm d}x)\propto{\bf 1}_{\mathbb{R}_{+}}(x)\exp(-cx^{\ell}){\rm d}x$
for some $k,\ell\geq0$. If $\ell<k$ then the integral $\int\pi({\rm d}x)\exp\left(M_{p}(x)/\alpha\right)$
in Proposition~\ref{prop:prod-avg-bound} is infinite. If $\ell>k$,
then the integral is finite. If $\ell=k$, then the integral is finite
if and only if $\alpha>b/c$.
\end{example}

\subsection{\label{subsec:Lognormal-example}Lognormal example}

We consider now a limiting case of the perturbations in a PMMH algorithm,
motivated by \cite[Theorem 1.1]{berard14}, which has also been analyzed
using other techniques \cite{doucet-2015,sherlock-2015}. The result
of \cite{berard14} concerns a particular mean 1 perturbation $W_{T,N}$
that is also a product of $T$ averages, with $N$ random variables
involved in each average, but where the random variables are not independent.
They show that, under regularity conditions, if $N=\alpha T$ there
is a $\sigma_{0}^{2}$ such that with $\sigma^{2}=\sigma_{0}^{2}/\alpha$,
$\log(W_{T,N})$ converges in distribution to $N(-\frac{1}{2}\sigma^{2},\sigma^{2})$
as $T\to\infty$.

We consider here the setting where for some large $T$, the log-perturbation
is exactly $N(-\frac{1}{2}\sigma^{2},\sigma^{2})$, in which case
one can think of $\sigma^{2}=\sigma_{0}^{2}T/N$, and since the precise
value of $T$ is not relevant we suppress it in the sequel. To be
explicit, we have that $W$ has law
\[
Q_{x}(\dif w)=\frac{1}{w\sigma\sqrt{2\pi}}\exp\left(-\frac{\left(\log w+\sigma^{2}/2\right)^{2}}{2\sigma^{2}}\right)\dif w,
\]
where $\dif w$ is Lebesgue measure on $\R_{+}$ and $\sigma>0$ is
a variance parameter, which we assume for simplicity is independent
of $x$. We note that a pseudo-marginal kernel with log-normal perturbations
can only converge subgeometrically, since the perturbations are not
bounded.

\subsubsection{Tail probabilities and convergence bound}
\begin{lem}
\label{lem:beta-s-log-normal}We have the bound, for $s>0$, $\tilde{\pi}_{x}(W\ge s)\le\beta(s)$,
where
\begin{equation}
\beta(s):=\exp\left(-\frac{\left(\left(\log s-\sigma^{2}/2\right)_{+}\right)^{2}}{2\sigma^{2}}\right).\label{eq:beta_lognormal_noCP}
\end{equation}
\end{lem}

\begin{rem}
Note that despite (\ref{eq:beta_lognormal_noCP}) being an upper bound
on the quantity of interest, it satisfies the conditions in Example~\ref{exa:cvg-super-to-strong}
and Remark~\ref{rem:pseudo-cv-marginal-general}, implying that the
rate of convergence of the marginal algorithm is recovered in the
limit, as $\sigma\rightarrow0$.
\end{rem}

Although it is theoretically possible to work directly with $\beta$
as in (\ref{eq:beta_lognormal_noCP}), in order to derive clean and
practically useful tuning guidelines, we now derive some tractable
bounds on the corresponding convergence rate.
\begin{lem}
\label{lem:boundK*}We have a lower bound on the convex conjugate,
for $0<v<1$, $K^{*}(v)\ge\frac{v}{2}\exp\left(-\sigma\sqrt{-2\log\frac{v}{2}}-\sigma^{2}/2\right)$.
\end{lem}

\begin{proof}
This is immediate from choosing $u=\exp\left(-\sigma\sqrt{-2\log\frac{v}{2}}-\sigma^{2}/2\right)$
in the definition of the convex conjugate, $K^{*}(v)=\sup_{u>0}\left\{ uv-u\beta(1/u)\right\} $.
\end{proof}
As before, we define $F(w):=\int_{w}^{1}\frac{\dif v}{K^{*}(v)}$.
We are able to deduce the following convergence bound.
\begin{lem}
\label{lem:log_normal_lambert}We have the upper bound for $x>0$,
\begin{equation}
F^{-1}(x)\le2\exp\left\{ -\frac{1}{2\sigma^{2}}\mathsf{W}^{2}\left(\frac{x\sigma^{2}}{2\exp(\sigma^{2}/2)}\right)\right\} ,\label{eq:lognormal_bound}
\end{equation}
where $\mathsf{W}^{2}$ denotes the Lambert function squared.
\end{lem}

\begin{prop}
Assume that the marginal chain $\bar{P}$ satisfies a strong Poincaré
inequality with constant $C_{\mathrm{P}}$ as in Remark~\ref{rem:strong_PI}.
Then the final convergence bound for the pseudo-marginal chain is
given by
\begin{equation}
F_{\mathrm{PM}}^{-1}(n)\le\frac{2}{C_{\mathrm{P}}}\exp\left\{ -\frac{1}{2\sigma^{2}}\mathsf{W}^{2}\left(\frac{C_{\mathrm{P}}n\sigma^{2}}{2\exp(\sigma^{2}/2)}\right)\right\} .\label{eq:log_normal_bound}
\end{equation}
\end{prop}

\begin{proof}
This is immediate from Corollary~\ref{cor:pseudo-mar_bound}, Lemma~\ref{lem:log_normal_lambert}.
\end{proof}

\subsubsection{Mixing times}

It is possible to obtain mixing time type results. 
\begin{prop}
\label{prop:mix-log-normal-basic}Let $\epsilon\in(0,1]$ and $\sigma^{2}>0$.
Then, to obtain $F_{{\rm PM}}^{-1}(n)\leq\epsilon^{2}$ it is sufficient
to take
\begin{equation}
n\geq\frac{2\sqrt{H(\epsilon)}}{C_{{\rm P}}\sigma}\exp\left(\frac{\sigma^{2}}{2}+\sqrt{H(\epsilon)}\sigma\right),\label{eq:n_bound-lognormal}
\end{equation}
where $H(\epsilon)=2\log(2/(\epsilon^{2}C_{{\rm P}}))$.
\end{prop}

\begin{proof}
One can calculate directly that the bound of $F_{{\rm PM}}^{-1}$
in (\ref{eq:lognormal_bound}) evaluated at the right-hand side of
(\ref{eq:n_bound-lognormal}) is equal to $2\exp(-H(\epsilon)/2)/C_{{\rm P}}$,
which combined with the definition of $H(\epsilon)$ and the monotonicity
of $F_{{\rm PM}}^{-1}$ gives the result.
\end{proof}
Now we consider the minimum computational budget required to achieve
a given precision of $\epsilon$, and the corresponding split between
the number of MCMC iterations $n$ and the number of particles $N$.
The budget required is significantly lower than the result in Proposition~\ref{prop:mix-log-normal-basic}
would imply for a fixed $N$ and therefore $\sigma^{2}$.
\begin{prop}
\label{prop:mixing-time-lognormal}Let $\epsilon\in(0,1]$. For simplicity,
let $\bar{n}$ and $\bar{N}$ be real-valued counterparts of $n$
and $N$, respectively. The `budget' function $(\bar{n},\bar{\sigma})\mapsto\mathsf{B}(\bar{n},\bar{\sigma})=\bar{n}\bar{N}=\bar{n}\sigma_{0}^{2}/\bar{\sigma}^{2}$
on $\mathbb{R}_{+}^{2}$ is minimized subject to the constraint $F_{{\rm PM}}^{-1}(\bar{n};\bar{\sigma})=\epsilon^{2}$
(with $F_{{\rm PM}}^{-1}$ as in (\ref{eq:lognormal_bound})) when
\[
\bar{\sigma}=\bar{\sigma}_{\star}(\epsilon):=\frac{\sqrt{H(\epsilon)+12}-\sqrt{H(\epsilon)}}{2},
\]
where $H(\epsilon):=2\log\bigl(\frac{2}{C_{{\rm P}}\epsilon^{2}}\bigr)>0$
and $\lim_{H(\epsilon)\to\infty}\sqrt{H(\epsilon)}\bar{\sigma}_{\star}(\epsilon)=3$.
Moreover, for $\epsilon>0$ such that $H(\epsilon)\geq1$, we obtain
$F_{{\rm PM}}^{-1}(\bar{n};\bar{\sigma})=\epsilon^{2}$ with $\bar{\sigma}(\epsilon)=3/\sqrt{H(\epsilon)}$,
\begin{align*}
\bar{N}(\epsilon) & =\frac{2}{9}\sigma_{0}^{2}\log\left(\tfrac{2}{C_{{\rm P}}\epsilon^{2}}\right),\\
\bar{n}(\epsilon) & \leq\frac{4\exp\left(15/2\right)}{3C_{{\rm P}}}\log\left(\tfrac{2}{C_{{\rm P}}\epsilon^{2}}\right),\\
\mathsf{B}(\epsilon) & \leq\frac{8\sigma_{0}^{2}\exp\left(15/2\right)}{27C_{{\rm P}}}\log\left(\tfrac{2}{C_{{\rm P}}\epsilon^{2}}\right)^{2},
\end{align*}
which is asymptotically accurate and optimal as $H(\epsilon)\to\infty$,
i.e. if $\epsilon\downarrow0$ or $C_{{\rm P}}\downarrow0$, except
that the constant factors $\exp\left(15/2\right)$ will tend to $\exp(3)$.
\end{prop}

These non-asymptotic results take into account both $C_{{\rm P}}$
and $\sigma_{0}^{2}$ in a natural manner and are easily interpretable.
We note that they also indicate how a given computational budget $\mathsf{B}$
should be split between $N$ and $n$ in order to achieve best precision:
in particular $N$ should increase as $\mathsf{B}$ increases. This
is to be contrasted with results (see \cite{doucet-2015,sherlock-2015}
and below) concerned with the asymptotic variance which recommend
a fixed number of particles for any $\mathsf{B}$ sufficiently large
and allocation of the remaining resources to iterating the MCMC algorithm
for this fixed number of particles.

\subsubsection{Asymptotic variance}

We now show that our bounds lead to recommendations for $N$ similar
to those of \cite{doucet-2015,sherlock-2015} when considering the
asymptotic variance as a criterion. We can use the bound (\ref{eq:log_normal_bound})
to give an upper bound on the resulting asymptotic variance.
\begin{lem}
\label{lem:bound-asymp-var}Fix a test function $f\in\ELL_{0}(\mu)$.
In the setting of Theorem~\ref{thm:rev_pos_conv}, for a reversible
Markov kernel $P$, the asymptotic variance $v(f,P)$ is bounded by
\[
v(f,P)\le-\|f\|_{2}^{2}+4\Phi(f)\sum_{n=0}^{\infty}F^{-1}(n).
\]
\end{lem}

\begin{example}
For our log-normal pseudo-marginal example, we can then ask for a
given $f$, how to tune $\sigma$ in order to minimize the resulting
bound on the asymptotic variance. We can bound
\begin{align*}
\sum_{n=1}^{\infty}F_{\mathrm{PM}}^{-1}(n) & \le\sum_{n=1}^{\infty}\frac{2}{C_{\mathrm{P}}}\exp\left\{ -\frac{1}{2\sigma^{2}}\mathsf{W}^{2}\left(\frac{C_{\mathrm{P}}n\sigma^{2}}{2\exp(\sigma^{2}/2)}\right)\right\} \\
 & \le\frac{2}{C_{\mathrm{P}}}\int_{0}^{\infty}\exp\left(-\mathsf{a}\mathsf{W}^{2}(\mathsf{b}x)\right)\,\dif x,
\end{align*}
where $\mathsf{a}:=1/(2\sigma^{2})$ and $\mathsf{b}:=C_{\mathrm{P}}\sigma^{2}/(2\exp(\sigma^{2}/2))$.
Here we used the fact that the Lambert function is increasing. Through
routine calculations and making use of the substitution $\mathsf{b}x=u\exp(u)\Leftrightarrow u=\mathsf{W}(\mathsf{b}x)$,
this integral can be simplified and written as
\[
\tilde{v}(\sigma):=\frac{1}{\mathsf{b}}\left[\exp(1/(4\mathsf{a}))\left(1+\frac{1}{2\mathsf{a}}\right)\mathsf{a}^{-1/2}\int_{-\mathsf{a}^{-1/2}/2}^{\infty}\exp(-w^{2})\,\dif w+\frac{1}{2\mathsf{a}}\right].
\]
In this final expression, both $\mathsf{a}$ and $\mathsf{b}$ depend
on $\sigma$, and the resulting function of $\sigma\mapsto\tilde{v}(\sigma)/\sigma^{2}$
can be optimized numerically, where we divide by $\sigma^{2}$ to
take into account the additional computational cost; see Figure~\ref{fig:logn_avar}.
Note that the optimal value $\sigma_{*}$ of $\sigma$ does not depend
on $C_{\mathrm{P}}$, and we find numerically that $\sigma_{*}\approx0.973$.
This is consistent with \cite{doucet-2015} who report optimal values
in the range $\sigma_{*}\approx1.0-1.7$ (dependent on the performance
of the marginal algorithm) using another bound on the asymptotic variance,
while \cite{sherlock-2015} find $\sigma_{*}\approx1.812$ using a
scaling and diffusion approximation.
\end{example}

\begin{figure}
\begin{centering}
\includegraphics[width=12cm,height=9cm]{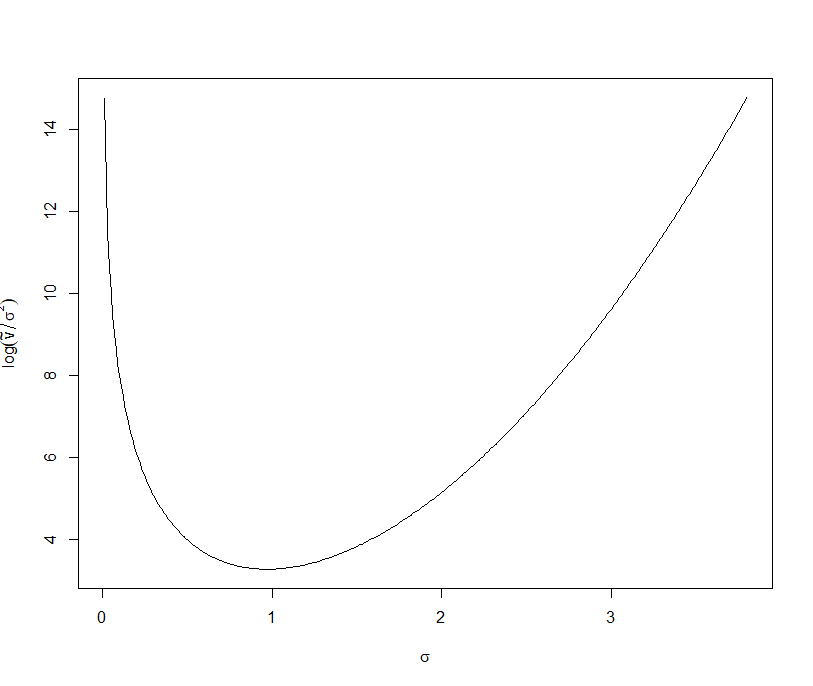}
\par\end{centering}
\caption{A plot of the function $\sigma\protect\mapsto\log\left(\tilde{v}(\sigma)/\sigma^{2}\right)$
in the case $C_{\mathrm{P}}=1$. \label{fig:logn_avar}}
\end{figure}

\appendix

\section{Proofs and other results for Section \ref{sec:Weak-Poincar=0000E9-inequalities-overview}
\label{app:rockner-wang}}

\subsection{Proofs for Section~\ref{subsec:General-case}}
\begin{lem}
\label{lem:properties-K-star}With $\beta$ as in Definition~\ref{def:Super-Poincar=0000E9-inequality},
let $K(u):=u\,\beta(1/u)$ for $u>0$ and $K(0):=0$. Define $K^{*}\colon[0,\infty)\rightarrow[0,\infty]$
to be the convex conjugate, $K^{*}(v):=\sup_{u\ge0}\left\{ uv-K(u)\right\} $
for $v\in[0,\infty)$. Let $K^{-}\left(v\right):=\sup\{u\geqslant0\colon K\left(u\right)\le uv\}$.
Write $a:=\sup\left\{ \frac{\left\Vert f\right\Vert _{2}^{2}}{\Phi\left(f\right)}:f\in\ELL_{0}\left(\mu\right)\setminus\left\{ 0\right\} \right\} $,
then
\begin{enumerate}
\item for $v\in\left[0,a\right)$, it holds that $K^{-}\left(v\right)\in\left[0,\infty\right)$,
and one can write 
\[
K^{*}\left(v\right)=\sup_{0\leq u\leq K^{-}(v)}\left\{ uv-K(u)\right\} ;
\]
\item $K^{*}(0)=0$ and for $v\neq0$, $K^{*}(v)>0$;
\item $v\mapsto K^{*}(v)$ is convex, continuous and strictly increasing
on its domain;
\item for $v\in\left[0,a\right]$, it holds that $K^{*}(v)\leq v$;
\item the function $v\mapsto v^{-1}K^{*}(v)$ is increasing.
\end{enumerate}
\end{lem}

\begin{proof}
To begin with, take $s\in\left[0,1\right)$, and compute that

\begin{align*}
\left\Vert f\right\Vert _{2}^{2} & \leqslant s\mathcal{E}\left(P^{*}P,f\right)+\beta\left(s\right)\Phi\left(f\right)\\
 & \leqslant s\left\Vert f\right\Vert _{2}^{2}+\beta\left(s\right)\Phi\left(f\right).
\end{align*}
Rearranging shows that for all $f\in\ELL_{0}\left(\mu\right)\setminus\left\{ 0\right\} $,
it holds that $\beta\left(s\right)\geqslant\left(\left\Vert f\right\Vert _{2}^{2}/\Phi\left(f\right)\right)\cdot\left(1-s\right)$,
and hence that

\begin{align*}
\beta\left(s\right) & \geqslant\left(\sup_{f\in\ELL_{0}\left(\mu\right)\setminus\left\{ 0\right\} }\frac{\left\Vert f\right\Vert _{2}^{2}}{\Phi\left(f\right)}\right)\cdot\left(1-s\right)\\
 & =a\cdot\left(1-s\right).
\end{align*}
It thus follows that $K\left(u\right)=u\beta\left(1/u\right)\geqslant a\cdot\left(u-1\right)$.
Now, write
\begin{align*}
K^{-}(v) & =\sup\{u\geqslant0\colon K\left(u\right)\le uv\}\\
 & \leqslant\sup\{u\geqslant0\colon a\cdot\left(u-1\right)\le uv\},
\end{align*}
to see that for $v\in\left[0,a\right)$, $K^{-}\left(v\right)<\infty$.
Now, for $u>K^{-}\left(v\right)$, it holds by definition that $K\left(u\right)>u\cdot v$,
and hence that $u\cdot v-K\left(u\right)<0$. Given that at $u=0$,
it holds that $u\cdot v-K\left(u\right)=0$, we can safely restrict
the supremum to be taken over the interval $\left[0,K^{-}\left(v\right)\right]$
as claimed. 

Write $K^{*}\left(0\right):=\sup_{u\ge0}\left\{ -K(u)\right\} $.
By nonnegativity of $K$, it is clear that $K^{*}\left(0\right)\leqslant0$,
and taking $u=0$ in the supremum ensures that $K^{*}\left(0\right)=0$.
For $v>0$, note that $u^{-1}K\left(u\right)$ tends to $0$ at $0$,
and so there exists some positive $u$ such that $K\left(u\right)<uv$,
from which one can deduce that $K^{*}\left(v\right)>0$. 

Convexity and continuity of $K^{*}$ follow from classical arguments.
Moreover, for $0<v<v'$, convexity implies that

\[
\frac{K^{*}\left(v\right)-K^{*}\left(0\right)}{v}\leqslant\frac{K^{*}\left(v'\right)-K^{*}\left(0\right)}{v'},
\]
from which one can deduce that $K^{*}$ is strictly increasing.

To upper-bound $K^{*}$, let $v\in\left[0,a\right]$, and compute
\begin{align*}
K^{*}\left(v\right) & =\sup_{u\ge0}\left\{ uv-K\left(u\right)\right\} \\
 & \leqslant\sup_{u\ge0}\left\{ uv-a\cdot\left(u-1\right)\right\} \\
 & =a+\sup_{u\geqslant0}\left\{ u\cdot\left(v-a\right)\right\} \\
 & =a,
\end{align*}
so for $v\in\left[0,a\right]$, it holds that $K^{*}\left(v\right)\leqslant a$.
Applying Jensen's inequality allows us to deduce that $K^{*}\left(v\right)\leqslant v$
on $\left[0,a\right]$.

The final point follows immediately from convexity of $K^{*}$ and
the fact that $K^{*}(0)=0$. 

\end{proof}
\begin{proof}[Proof of Lemma~\ref{lem:F_a_properties}]
Note that $x\mapsto F_{a}(x)$ is well-defined since $K^{*}$ is
continuous and $K^{*}(v)>0$ and clearly strictly decreasing with
$F_{a}\left(a\right)=0$.\textcolor{black}{{} Convexity of $F_{a}$
follows from monotonicity of $K^{*}$.} Further, $\lim_{x\downarrow0}F_{a}(x)=\infty$
since $1/K^{*}(v)\geq v^{-1}$ for $0<v<a$ as in Lemma~\ref{lem:properties-K-star}.
As such, there is a well-defined inverse function $F_{a}^{-1}:(0,\infty)\to(0,a)$,
with $F_{a}^{-1}(x)\to0$ as $x\to\infty$.
\end{proof}
\begin{proof}[Proof of Theorem \ref{thm:WPI_F_bd}.]
 Since $v\mapsto1/K^{*}(v)$ is decreasing and $k\mapsto\|P^{k}f\|_{2}$
decreasing, for $n\geq1$,
\begin{align*}
F_{a}\big(\|P^{n}f\|_{2}^{2}/\Phi(f)\big)-F_{a}\big(\|P^{n-1}f\|_{2}^{2}/\Phi(f)\big) & =\int_{\|P^{n}f\|_{2}^{2}/\Phi(f)}^{\|P^{n-1}f\|_{2}^{2}/\Phi(f)}1/K^{*}(v)\,{\rm d}v\\
 & \geq\frac{\|P^{n-1}f\|_{2}^{2}-\|P^{n}f\|_{2}^{2}}{K^{*}\big(\|P^{n-1}f\|_{2}^{2}/\Phi(f)\big)\Phi(f)}\\
 & =\frac{\mathcal{E}(P^{*}P,P^{n-1}f)/\Phi(f)}{K^{*}\big(\|P^{n-1}f\|_{2}^{2}/\Phi(f)\big)}.
\end{align*}
Now, applying the optimised \textcolor{black}{weak} Poincaré inequality
with $f$ replaced by $P^{n-1}f$, one can bound

\[
\frac{\mathcal{E}\left(P^{*}P,P^{n-1}f\right)}{\Phi\left(P^{n-1}f\right)}\geqslant K^{*}\left(\frac{\left\Vert P^{n-1}f\right\Vert _{2}^{2}}{\Phi\left(P^{n-1}f\right)}\right).
\]
By assumption, $\Phi\left(P^{n-1}f\right)\leqslant\Phi\left(f\right)$,
and hence $\frac{\left\Vert P^{n-1}f\right\Vert _{2}^{2}}{\Phi\left(P^{n-1}f\right)}\geqslant\frac{\left\Vert P^{n-1}f\right\Vert _{2}^{2}}{\Phi\left(f\right)}$.
Additionally, since $u\mapsto u^{-1}K^{*}\left(u\right)$ is increasing
(Lemma~\ref{lem:properties-K-star}), we obtain the conclusion

\begin{align*}
\left(\frac{\left\Vert P^{n-1}f\right\Vert _{2}^{2}}{\Phi\left(P^{n-1}f\right)}\right)^{-1}\cdot K^{*}\left(\frac{\left\Vert P^{n-1}f\right\Vert _{2}^{2}}{\Phi\left(P^{n-1}f\right)}\right) & \geqslant\left(\frac{\left\Vert P^{n-1}f\right\Vert _{2}^{2}}{\Phi\left(f\right)}\right)^{-1}\cdot K^{*}\left(\frac{\left\Vert P^{n-1}f\right\Vert _{2}^{2}}{\Phi\left(f\right)}\right)\\
\implies K^{*}\left(\frac{\left\Vert P^{n-1}f\right\Vert _{2}^{2}}{\Phi\left(P^{n-1}f\right)}\right) & \geqslant\left(\frac{\Phi\left(f\right)}{\Phi\left(P^{n-1}f\right)}\right)\cdot K^{*}\left(\frac{\left\Vert P^{n-1}f\right\Vert _{2}^{2}}{\Phi\left(f\right)}\right).
\end{align*}
We thus see that
\begin{align*}
\frac{\mathcal{E}\left(P^{*}P,P^{n-1}f\right)}{\Phi\left(P^{n-1}f\right)} & \geqslant\frac{\Phi\left(f\right)}{\Phi\left(P^{n-1}f\right)}\cdot K^{*}\left(\frac{\left\Vert P^{n-1}f\right\Vert _{2}^{2}}{\Phi\left(f\right)}\right)\\
\implies\frac{\mathcal{E}\left(P^{*}P,P^{n-1}f\right)}{\Phi\left(f\right)} & \geqslant K^{*}\left(\frac{\left\Vert P^{n-1}f\right\Vert _{2}^{2}}{\Phi\left(f\right)}\right).
\end{align*}
Combining this with our earlier inequalities, we see that 
\[
{\color{red}{\color{black}F_{a}\big(\|P^{n}f\|_{2}^{2}/\Phi(f)\big)-F_{a}\big(\|P^{n-1}f\|_{2}^{2}/\Phi(f)\big)\geqslant1.}}
\]
As a result, for $n\geq1$, it holds that
\[
F_{a}\big(\|P^{n}f\|_{2}^{2}/\Phi(f)\big)-F_{a}\big(\|f\|_{2}^{2}/\Phi(f)\big)\geq n,
\]
from which we obtain
\[
\|P^{n}f\|_{2}^{2}\leq\Phi\left(f\right)F_{a}^{-1}(n).
\]
\end{proof}
Weak Poincaré inequalities were considered in \cite{Rockner2001},
who used it to derive subexponential convergence rates for continuous-time
Markov semigroups. We give here the analogous discrete-time result,
and note that the obtained rate is weaker than our previous Theorem~\ref{thm:WPI_F_bd}.
\begin{prop}
Assume that a weak Poincaré inequality for $P^{*}P$ holds as in Definition~\ref{def:Weak-Poincar=0000E9-inequality}.
Then, for any $f\in\ELL_{0}(\mu)$ with $\Phi(f)<\infty$, we have
the bound
\[
\|P^{n}f\|_{2}^{2}\le\tilde{\gamma}(n)\left(\|f\|_{2}^{2}+\Phi(f)\right),
\]
where
\begin{equation}
\tilde{\gamma}(n)=\inf\left\{ r>0\colon\big(1-1/\alpha(r)\big)^{n}\le r\right\} ,\label{eq:beta(n)}
\end{equation}
and $\tilde{\gamma}$ satisfies $\gamma(n)\downarrow0,$ as $n\to\infty$,
replacing $\alpha$ with $\alpha\vee1$ if necessary.
\end{prop}

\begin{proof}
We will use the following identity; for any $g\in\ELL(\mu)$,
\begin{align*}
\|Pg\|_{2}^{2}-\|g\|_{2}^{2} & =\bigl\langle P^{*}Pg,g\bigr\rangle-\bigl\langle g,g\bigr\rangle\\
 & =-\mathcal{E}(P^{*}P,g).
\end{align*}
This implies in particular for $n\geq1$,
\begin{align*}
\|P^{n}f\|_{2}^{2}-\|P^{n-1}f\|_{2}^{2} & =-\mathcal{E}(P^{*}P,P^{n-1}f)\\
 & \leq-\|P^{n-1}f\|_{2}^{2}/\alpha(r)+(r/\alpha(r))\Phi(P^{n-1}f).
\end{align*}
Therefore with $u_{n}:=\|P^{n}f\|_{2}^{2}$ and using the nonexpansive
property of $\Phi$, we have
\begin{equation}
u_{n}\leq\big(1-1/\alpha(r)\big)u_{n-1}+(r/\alpha(r))\Phi(f),\label{eq:un_iteration}
\end{equation}
 and by iterating this, we obtain
\begin{align*}
u_{n} & \le\big(1-1/\alpha(r)\big)^{n}u_{0}+r\Phi(f).
\end{align*}
Thus if we define $\tilde{\gamma}(n)$ as in (\ref{eq:beta(n)}),
the desired bound on $\|P^{n}f\|_{2}^{2}$ is immediate.
\end{proof}
\begin{rem}
The rate (\ref{eq:beta(n)}) is weaker than the rate of Theorem~\ref{thm:WPI_F_bd}
because in the proof of Theorem~\ref{thm:WPI_F_bd}, we are essentially
optimizing over $r>0$ at every step of the iteration (\ref{eq:un_iteration}),
rather than fixing the same $r>0$ for every iteration and then optimizing
only at the end as in (\ref{eq:beta(n)}). It is noted in \cite[Corollary 2.4]{Rockner2001},
through an additional iterative argument, it is possible in some cases
to recover improved rates. This is not necessary for us, since our
Theorem~\ref{thm:WPI_F_bd} automatically returns the improved rates
obtained by \cite{Rockner2001} for their examples.
\end{rem}

\begin{proof}[Proof of Lemma~\ref{lem:rescaling-beta-rescaling-invF}]
For the first part, notice that $u[v-c_{1}\beta(c_{2}/u)]=c_{1}c_{2}(u/c_{2})\{v/c_{1}-\beta(c_{2}/u)\}$.
Then,
\[
{\color{red}{\color{black}\tilde{F}_{a}(w)=\int_{w}^{a}\frac{{\rm d}v}{c_{1}c_{2}K^{*}(v/c_{1})}=\frac{1}{c_{2}}\int_{w/c_{1}}^{a/c_{1}}\frac{c_{1}{\rm d}v'}{c_{1}K^{*}(v')}=c_{2}^{-1}F_{a/c_{1}}(w/c_{1}),}}
\]
and the result follows.
\end{proof}
\begin{proof}[Proof of Lemma \ref{lem:rate-beta-decays-polynomial}.]
 From direct calculation, we find that
\[
K^{*}(v)=C\left(c_{0},c_{1}\right)v^{1+c_{1}^{-1}},
\]
where $C\left(c_{0},c_{1}\right):=\frac{c_{0}c_{1}}{\left(c_{0}(1+c_{1})\right)^{1+c_{1}^{-1}}}$.
We then calculate that, for a fixed $a>0$,
\[
F(v)=\int_{w}^{a}\frac{\dif v}{K^{*}(v)}=\tilde{C}\left(c_{0},c_{1}\right)\left[w^{-c_{1}^{-1}}-a^{-c_{1}^{-1}}\right],
\]
with $\tilde{C}\left(c_{0},c_{1}\right):=c_{1}/C\left(c_{0},c_{1}\right)=(1+c_{1})^{1+c_{1}^{-1}}c_{0}^{c_{1}^{-1}}$.
Inverting this, we find that
\begin{align*}
F^{-1}(n) & =\left(\frac{1}{n/\tilde{C}\left(c_{0},c_{1}\right)+a^{-c_{1}^{-1}}}\right)^{c_{1}}\\
 & \le\tilde{C}\left(c_{0},c_{1}\right)^{c_{1}}n^{-c_{1}}\\
 & =c_{0}(1+c_{1})^{1+c_{1}}n^{-c_{1}}.
\end{align*}
\end{proof}
\begin{proof}[Proof of Lemma~\ref{lem:rate-beta-decays-exponentially}.]
 We have
\begin{align*}
K^{*}(v) & =\sup_{u\in[0,\infty)}\left\{ uv-u\eta_{0}\exp\left(-\eta_{1}u^{-\eta_{2}}\right)\right\} .
\end{align*}
For some $\eta>0$ and $v\in(0,1]$, take \textbf{$u_{v}:=\left(\frac{1}{\eta}\log\left(\frac{1}{v}\right)\right)^{-1/\eta_{2}}$},
that is such that \textbf{$v=\exp\left(-\eta u_{v}^{-\eta_{2}}\right)<1$}.
As a result we have a lower bound on $K^{*}$ as
\[
K^{*}(v)\geqslant v\left(\frac{1}{\eta}\log\left(\frac{1}{v}\right)\right)^{-1/\eta_{2}}-\eta_{0}v^{\eta_{1}/\eta}.
\]
Provided that $\eta\text{\ensuremath{\in\left(0,\eta_{1}\right)}}$,
the second term will decay faster as $v\downarrow0$ and the first
statement follows. 

The second statement follows upon noticing that for $0\le w\le v_{0}$,
and writing $M:=\int_{v_{0}}^{a}\frac{\dif v}{K^{*}(v)}$,
\begin{align*}
F_{a}(w) & \leq C^{-1}\int_{w}^{v_{0}}v^{-1}\left(-\log v\right)^{1/\eta_{2}}{\rm d}v+M\\
 & =C^{-1}\frac{\eta_{2}}{1+\eta_{2}}\left[\left(-\log w\right)^{(\eta_{2}+1)/\eta_{2}}-\left(-\log v_{0}\right)^{(\eta_{2}+1)/\eta_{2}}\right]+M\\
 & \le C^{-1}\frac{\eta_{2}}{1+\eta_{2}}\left(-\log w\right)^{(\eta_{2}+1)/\eta_{2}}+M.
\end{align*}
This leads for $n\geq M$ to
\[
F_{a}^{-1}(n)\leq\exp\left(-\left(C\frac{1+\eta_{2}}{\eta_{2}}(n-M)\right)^{\eta_{2}/(1+\eta_{2})}\right).
\]
For $n\geq M$, by concavity of $[0,M]\ni x\mapsto(n-x)^{\gamma}$
for $\gamma\in(0,1)$ we have $n^{\gamma}-\gamma(n-M)^{\gamma-1}M\leq(n-M)^{\gamma}$
and $(n-M)^{\gamma-1}\leq(\lceil M\rceil-M)^{\gamma-1}$ so we conclude
that there exists $C'>0$ such that for $n\in\mathbb{N}$,
\[
F_{a}^{-1}(n)\leq C'\exp\left(-\left(C\frac{1+\eta_{2}}{\eta_{2}}n\right)^{\eta_{2}/(1+\eta_{2})}\right).
\]
\end{proof}
\begin{proof}[Proof of Lemma~\ref{lem:rate-beta-decays-log}.]
\textcolor{black}{We have}

\textcolor{black}{
\begin{align*}
K^{*}\left(v\right) & =\sup_{u\in[0,\infty)}\left\{ uv-uc_{0}\cdot\left(\log\max\left(c_{1},\frac{1}{u}\right)\right)^{-p}\right\} \\
 & =\max\left\{ \sup_{u\in\left[0,c_{1}^{-1}\right]}\left\{ uv-uc_{0}\cdot\left(\log\left(\frac{1}{u}\right)\right)^{-p}\right\} ,\sup_{u\geqslant c_{1}^{-1}}\left\{ uv-uc_{0}\log c_{1}\right\} \right\} .
\end{align*}
For $v<v_{0}:=\left(\log c-1\right)^{-1/p}$, take $u=\exp\left(-\left(1+\left(v/c_{0}\right)^{-1/p}\right)\right)$
to write}

\textcolor{black}{
\begin{align*}
K^{*}\left(v\right) & \geqslant\exp\left(-\left(1+\left(v/c_{0}\right)^{-1/p}\right)\right)\cdot v-c_{0}\cdot\exp\left(-\left(1+\left(v/c_{0}\right)^{-1/p}\right)\right)\cdot\left(1+\left(v/c_{0}\right)^{-1/p}\right)^{-p}\\
 & =\exp\left(-\left(1+\left(v/c_{0}\right)^{-1/p}\right)\right)\cdot\left\{ v-c_{0}\cdot\left(1+\left(v/c_{0}\right)^{-1/p}\right)^{-p}\right\} \\
 & =\exp\left(-\left(1+\left(v/c_{0}\right)^{-1/p}\right)\right)\cdot\left\{ v-c_{0}\cdot\frac{v/c_{0}}{\left(1+\left(v/c_{0}\right)^{1/p}\right)^{p}}\right\} \\
 & =v\cdot\exp\left(-\left(1+\left(v/c_{0}\right)^{-1/p}\right)\right)\cdot\left\{ 1-\left(1+\left(v/c_{0}\right)^{1/p}\right)^{-p}\right\} \\
 & \geqslant v\cdot\exp\left(-\left(1+\left(v/c_{0}\right)^{-1/p}\right)\right)\cdot p\cdot\left(v/c_{0}\right)^{1/p}\\
 & =C\cdot v^{1+1/p}\cdot\exp\left(-\left(v/c_{0}\right)^{-1/p}\right).
\end{align*}
Thus, we can bound
\begin{align*}
F_{a}\left(w\right) & =\int_{w}^{a}\frac{{\rm d}v}{K^{*}\left(v\right)}\\
 & =\int_{w}^{v_{0}}\frac{{\rm d}v}{K^{*}\left(v\right)}+\int_{v_{0}}^{a}\frac{{\rm d}v}{K^{*}\left(v\right)}\\
 & \leqslant C^{-1}\cdot\int_{w}^{v_{0}}\frac{{\rm d}v}{v^{1+1/p}\cdot\exp\left(-\left(v/c_{0}\right)^{-1/p}\right)}+C^{'}\\
 & =C^{-1}\cdot\int_{\left(v_{0}/c_{0}\right)^{-1/p}}^{\left(w/c_{0}\right)^{-1/p}}c_{0}^{-1/p}\cdot p\cdot\exp z\,{\rm d}z+C^{'}\\
 & \leqslant C^{''}\cdot\exp\left(\left(w/c_{0}\right)^{-1/p}\right)+C^{'},
\end{align*}
and invert to deduce that
\[
F_{a}^{-1}\left(n\right)\leqslant c_{0}\cdot\left(\log\left(\frac{n-C^{'}}{C^{''}}\right)\right)^{-p}.
\]
Recalling that $F_{a}^{-1}$ is a priori bounded from above by $a<\infty$,
one then concludes that for some constant $C^{'''}>0$, it holds that
$F_{a}^{-1}\left(n\right)\leqslant C^{'''}\cdot\log\max\left(n,2\right)^{-p}$,
as claimed.}
\end{proof}

\subsection{Proofs for Section~\ref{subsec:Reversible-case}}
\begin{proof}[Proof of Lemma~\ref{lem:P-to-P2-gap}]
From routine calculations, the conclusion is equivalent to
\[
(1-c_{\mathrm{gap}})\langle f,f\rangle\ge\langle(P^{2}-c_{\mathrm{gap}}P)f,f\rangle.
\]
By the spectral theorem \cite[Chapter 9]{Conway90}, this can be expressed
as
\begin{equation}
(1-c_{\mathrm{gap}})\int_{-1}^{1}\dif\nu_{f}(\lambda)\ge\int_{-1}^{1}(\lambda^{2}-c_{\mathrm{gap}}\lambda)\,\dif\nu_{f}(\lambda).\label{eq:lem15_eq}
\end{equation}
It is easily seen that for any $0<c\le1$ we have that, $\lambda^{2}-c\lambda\le1-c$
for all $\lambda\ge-1+c$. By choosing $c=c_{\mathrm{gap}}$, and
see that equation (\ref{eq:lem15_eq}) holds for this choice of $c$
under the assumption of a left spectral gap.
\end{proof}
\textcolor{black}{To prove Theorem~\ref{prop:left-WPI-P}, we first
need the following technical lemma, which is taken from \cite[Theorem VI.9]{ReedSimon80}:}
\begin{lem}
\textcolor{black}{For any bounded positive linear operator $A:\mathcal{H}\to\mathcal{H}$
on a Hilbert space $\mathcal{H}$, there is a unique bounded positive
linear operator $B:\mathcal{H}\to\mathcal{H}$ with $B^{2}=A$, which
is realised as the limit (in operator norm) $B=\Id+\sum_{n=1}^{\infty}c_{n}(\Id-A)^{n}$,
where the constants $(c_{n})$ are known explicitly and the series
$\sum c_{n}$ converges absolutely. Furthermore, $B$ commutes with
every bounded operator which commutes with $A$.\label{lem:sqrt_lemma}}
\end{lem}

\begin{proof}[Proof of Theorem~\ref{prop:left-WPI-P}]
\textcolor{black}{Firstly, since $P$ is a reversible Markov transition
kernel, we have that $\Id+P\ge0$ as an operator on $\ELL_{0}(\mu)$.
Thus by Lemma~(\ref{lem:sqrt_lemma}), we can define a square root
operator $(\Id+P)^{1/2}$, which also commutes with $(\Id-P)$. Thus
we have the representation $(\Id-P^{2})=(\Id-P)(\Id+P)=(\Id+P)^{1/2}(\Id-P)(\Id+P)^{1/2}$. }

\textcolor{black}{Let us now fix some $f\in\ELL_{0}(\mu)$, and we
set $g:=(\Id+P)^{1/2}f$. We have 
\begin{equation}
\langle f,(\Id+P)f\rangle=\|(\Id+P)^{1/2}f\|_{2}^{2}=\|g\|_{2}^{2}.\label{eq:f-gnorm}
\end{equation}
Note that $g\in\ELL_{0}(\mu)$, since from Lemma~\ref{lem:sqrt_lemma},
we have the representation $g=(\Id+P)^{1/2}f=\Id+\sum_{k=1}^{\infty}c_{k}P^{k}f$
for some known absolutely convergent series $(c_{k})$; thus by $\mu$-invariance
of $P$, $\mu(g)=0$ since $\mu(f)=0$. Thus we can write,
\begin{align*}
\langle f,(\Id-P^{2})f\rangle & =\langle f,(\Id+P)^{1/2}(\Id-P)(\Id+P)^{1/2}f\rangle\\
 & =\langle g,(\Id-P)g\rangle.
\end{align*}
Now, since $g\in\ELL_{0}(\mu)$, we have from (\ref{eq:beta+}) that
for any $s>0$,
\begin{equation}
\|g\|_{2}^{2}\le s\langle g,(\Id-P)g\rangle+\beta_{+}(s)\Phi(g).\label{eq:WPI_g}
\end{equation}
Furthermore we have from (\ref{eq:f-gnorm}) and (\ref{eq:beta-})
that for any $s>0$,
\begin{equation}
\|f\|_{2}^{2}\le s\|g\|^{2}+\beta_{-}(s)\Phi(f).\label{eq:WPI_fg}
\end{equation}
Thus combining (\ref{eq:WPI_g}) and (\ref{eq:WPI_fg}), we find for
any $s>0$, and $s_{1}>0$, $s_{2}>0$ with $s_{1}s_{2}>0$,
\begin{align*}
\|f\|_{2}^{2} & \le\beta_{-}(s_{1})\Phi(f)+s_{1}s_{2}\langle g,(\Id-P)g\rangle+s_{1}\beta_{+}(s_{2})\Phi(g)\\
 & =s\langle f,(\Id-P^{2})f\rangle+\beta_{-}(s_{1})\Phi(f)+s_{1}\beta_{+}(s_{2})\Phi\left((\Id+P)^{1/2}f\right)\\
 & \le s\langle f,(\Id-P^{2})f\rangle+\left(\beta_{-}(s_{1})+s_{1}\beta_{+}(s_{2})\right)\left[\Phi(f)\vee\Phi\left((\Id+P)^{1/2}f\right)\right].
\end{align*}
Thus taking an infimum over $s_{1},s_{2}$, we have arrived at (\ref{eq:P2_WPI_beta}).}

\textcolor{black}{Finally, we check that $\beta$ and $\tilde{\Phi}$
satisfy the necessary conditions as in Definition~\ref{def:WPI_rev}.
The conditions for $\beta$ are established in the proof of Theorem~\ref{thm:chaining-with-WPI}.
For $\tilde{\Phi}$, the fact that $\tilde{\Phi}(cf)=c^{2}\tilde{\Phi}(f)$
is clear, and we have
\begin{align*}
\Phi\left((\Id+P)^{1/2}P^{n}f\right) & =\Phi\left(P^{n}(\Id+P)^{1/2}f\right)\\
 & \le\Phi\left((\Id+P)^{1/2}f\right),
\end{align*}
as desired. Finally, since $\Phi\le\tilde{\Phi}$, we immediately
have that for $f\in\ELL_{0}(\mu)$, $\|f\|_{2}^{2}\le a\tilde{\Phi}(f)$.}
\end{proof}
\begin{proof}[Proof of Proposition~\ref{prop:necessity-reversible-scenario}]
The proof borrows ideas from \cite{Rockner2001}. Recall that a\textcolor{black}{{}
weak }Poincaré inequality holds if we can find a function $\beta$
such that for all $f$ satisfying $\Phi\left(f\right)<\infty$, it
holds that $\|f\|_{2}^{2}\leq s\mathcal{E}(P^{*}P,f)+\beta(s)\Phi(f)$.
The sharpest such $\beta$ can then be recovered as 

\[
\beta\left(s\right):=\sup_{f:0<\Phi\left(f\right)<\infty}\left\{ \frac{\|f\|_{2}^{2}}{\Phi(f)}-s\frac{\mathcal{E}(P^{*}P,f)}{\Phi\left(f\right)}\right\} .
\]
Defining $u\left(f\right):=\frac{\|f\|_{2}^{2}}{\Phi(f)},v\left(f\right):=\frac{\mathcal{E}(P^{*}P,f)}{\Phi\left(f\right)}$,
we thus seek to find uniform upper bounds on the scale-free quantity
$u\left(f\right)-s\cdot v\left(f\right)$ over the set of such $f$.

By self-adjointness of $P$, for any $n\geqslant0$ we can apply the
spectral theorem to write

\[
P^{n}=\int_{\sigma\left(P\right)}\lambda^{n}\,\dif E_{\lambda},
\]
where $\sigma\left(P\right)$ is the spectrum of $P$ and $\left\{ E_{\lambda}:\lambda\in\sigma\left(P\right)\right\} $
is the corresponding spectral family. For $f\in\ELL_{0}\left(\mu\right)$
such that $\|f\|_{2}^{2}=1$, write
\[
\|P^{n}f\|_{2}^{2}=\int_{\sigma\left(P\right)}\left|\lambda\right|^{2n}\,\dif\left\Vert E_{\lambda}f\right\Vert ^{2}\geqslant\left(\int_{\sigma\left(P\right)}\left|\lambda\right|^{2}\,\dif\left\Vert E_{\lambda}f\right\Vert ^{2}\right)^{n}=\left\Vert Pf\right\Vert _{2}^{2n},
\]
noting that $\dif\left\Vert E_{\lambda}f\right\Vert ^{2}$ is a probability
measure and applying Jensen's inequality. By a scaling argument, one
can then deduce that $\left\Vert Pf\right\Vert _{2}\leqslant\left\Vert f\right\Vert _{2}^{1-1/n}\cdot\left\Vert P^{n}f\right\Vert _{2}^{1/n}$.

Combining this with our assumption, we obtain the estimate $\left\Vert Pf\right\Vert _{2}^{2}\leqslant\left\Vert f\right\Vert _{2}^{2\left(1-1/n\right)}\cdot\left(\gamma\left(n\right)\cdot\Phi\left(f\right)\right)^{1/n}$.
Rearranging this, we see that

\begin{align*}
\mathcal{E}\left(P^{*}P,f\right) & =\|f\|_{2}^{2}-\left\Vert Pf\right\Vert _{2}^{2}\geqslant\|f\|_{2}^{2}-\left\Vert f\right\Vert _{2}^{2\left(1-1/n\right)}\cdot\left(\gamma\left(n\right)\cdot\Phi\left(f\right)\right)^{1/n},
\end{align*}
and dividing through by $\Phi\left(f\right)$ tells us that $v\left(f\right)\geqslant u\left(f\right)-\gamma\left(n\right)^{1/n}u\left(f\right)^{1-1/n}$.
One can then write that

\begin{align*}
u\left(f\right)-s\cdot v\left(f\right) & \leqslant s\cdot\gamma\left(n\right)^{1/n}u\left(f\right)^{1-1/n}-\left(s-1\right)\cdot u\left(f\right)\\
 & \leqslant\frac{s^{n}}{\left(s-1\right)^{n-1}}\cdot\frac{\left(n-1\right)^{n-1}}{n^{n}}\cdot\gamma\left(n\right),
\end{align*}
where the second inequality comes from taking the supremum of the
right-hand side over $u\left(f\right)>0$, provided that $n\geqslant2$.
One can then deduce that

\[
\beta\left(s\right)=\sup_{f:0<\Phi\left(f\right)<\infty}\left\{ u\left(f\right)-s\cdot v\left(f\right)\right\} \leqslant\inf_{n\geq2}\left\{ \frac{s^{n}}{\left(s-1\right)^{n-1}}\cdot\frac{\left(n-1\right)^{n-1}}{n^{n}}\cdot\gamma\left(n\right)\right\} =:\beta_{0}\left(s\right).
\]

Defining $\beta_{1}\left(s\right):=\sup_{t\geqslant s}\beta_{0}\left(t\right)$,
one can see that $\beta_{1}$ is also a valid upper bound, and is
decreasing. Moreover, writing $s=n+\delta$ with $n\in\text{\ensuremath{\mathbb{N}}},\delta\in\left[0,1\right)$,
one can bound

\begin{align*}
\beta_{0}\left(s\right) & \leqslant\frac{s^{n}}{\left(s-1\right)^{n-1}}\cdot\frac{\left(n-1\right)^{n-1}}{n^{n}}\cdot\gamma\left(n\right)\\
 & =\left(\frac{n+\delta}{n}\right)^{n}\cdot\left(\frac{n-1}{n-1+\delta}\right)^{n-1}\cdot\gamma\left(n\right)\\
 & \leqslant\exp\left(\delta\right)\cdot1\cdot\gamma\left(n\right)\\
 & \leqslant e\cdot\gamma\left(\left\lfloor s\right\rfloor \right),
\end{align*}
which vanishes as $s$ grows. Moreover, since $\gamma$ is decreasing,
one can write

\[
\beta_{1}\left(s\right)=\sup_{t\geqslant s}\beta_{0}\left(t\right)\leqslant e\cdot\sup_{t\geqslant s}\gamma\left(\left\lfloor t\right\rfloor \right)=e\cdot\gamma\left(\left\lfloor s\right\rfloor \right),
\]
to conclude that $\beta_{1}$ decreases to $0$.

\textcolor{black}{Assume now that our a priori bound has the form}

\textcolor{black}{
\[
\left\Vert P^{n}f\right\Vert _{2}^{2}\leqslant\Phi\left(f\right)\cdot F^{-1}\left(n+F\left(\frac{\left\Vert f\right\Vert _{2}^{2}}{\Phi\left(f\right)}\right)\right)
\]
with $F$ decreasing, continuous, and blowing up at $0$, $F^{-1}$
decreasing, continuous, and convex, and $\log\left(-\mathrm{D}F^{-1}\right)$
convex. Taking $n=1$, rearrangement shows that}

\textcolor{black}{
\[
\frac{\mathcal{E}\left(P^{*}P,f\right)}{\Phi\left(f\right)}\geqslant v-F^{-1}\left(1+F\left(v\right)\right),
\]
with $v=\frac{\left\Vert f\right\Vert _{2}^{2}}{\Phi\left(f\right)}$.
As such, it suffices to prove that the mapping $K_{1}^{*}:v\mapsto v-F^{-1}\left(1+F\left(v\right)\right)$
is $0$ at $0$, increasing, and convex.}

\textcolor{black}{For the first statement, since $F^{-1}$ is decreasing
and non-negative $0\leq F^{-1}\left(1+F\left(v\right)\right)\leq F^{-1}\left(F\left(v\right)\right)=v$,
we deduce}

\textcolor{black}{
\begin{align*}
\lim_{v\to0^{+}}\left\{ v-F^{-1}\left(1+F\left(v\right)\right)\right\}  & =0.
\end{align*}
}

\textcolor{black}{For the second statement, compute the derivative
of $K_{1}^{*}$:}

\textcolor{black}{
\[
\left(\mathrm{D}K_{1}^{*}\right)\left(v\right)=1-\frac{\left(\mathrm{D}F^{-1}\right)\left(1+F\left(v\right)\right)}{\left(\mathrm{D}F^{-1}\right)\left(F\left(v\right)\right)}.
\]
Noting that $F^{-1}$ is decreasing and convex, it follows that the
fractional term is less than $1$, from which the claim follows.}

\textcolor{black}{Finally, for convexity, observe that the mapping
$H_{1}:v\mapsto\log\left(1-\left(\mathrm{D}K_{1}^{*}\right)\left(v\right)\right)$
is a monotone decreasing function of $\mathrm{D}K_{1}^{*}$, and hence
it suffices to show that $H_{1}$ is decreasing. By assumption, the
mapping $L:v\mapsto\log\left(\left(-\mathrm{D}F^{-1}\right)\left(v\right)\right)$
is convex, and $F$ is decreasing, from which it follows by inspection
that $H_{1}=L\circ\left(1+F\right)-L\circ F$ is decreasing, as required.
We thus deduce convexity of $K_{1}^{*}$.}

\textcolor{black}{Under the assumption that our a priori bound takes
the form}

\textcolor{black}{
\[
\frac{\left\Vert P^{n}f\right\Vert _{2}^{2}}{\Phi\left(f\right)}\leqslant\left(\mathrm{Id}-\tilde{K}^{*}\right)^{\circ n}\left(\frac{\left\Vert f\right\Vert _{2}^{2}}{\Phi\left(f\right)}\right),
\]
we can take $n=1$ as before and rearrange to directly deduce that
the desired WPI holds.}
\end{proof}

\subsubsection*{Calculations for Remark~\textcolor{red}{\ref{rem:nec-comment}}}

\textcolor{black}{In Remark~\ref{rem:nec-comment}, it is mentioned
that for the examples in Section~\ref{subsec:Examples-of-b}, the
upper bounds which are implied by that remark are sufficient for recovering
a closely related $\beta$. The relevant calculations are supplied
here.}

\textcolor{black}{Recall the bound
\[
\beta\left(s\right)\leqslant\sup_{t\geqslant s}\inf_{n\geqslant2}\left\{ \gamma\left(n\right)\cdot\left(\frac{n}{t-1}\right)^{-1}\cdot\exp\left(\frac{n}{t-1}\right)\right\} .
\]
In the case where $\gamma\left(n\right)=\left(\log\left(n+c\right)\right)^{-p}$
for some $c>1$, write}

\textcolor{black}{
\begin{align*}
\beta\left(s\right) & \leqslant\sup_{t\geqslant s}\inf_{n\geqslant2}\left\{ \left(\log\left(n+c\right)\right)^{-p}\cdot\left(\frac{n}{t-1}\right)^{-1}\cdot\exp\left(\frac{n}{t-1}\right)\right\} \\
 & \leqslant\sup_{t\geqslant s}\left\{ \left(\log\left(\left\lceil t\right\rceil -1+c\right)\right)^{-p}\cdot\left(\frac{\left\lceil t\right\rceil -1}{t-1}\right)^{-1}\cdot\exp\left(\frac{\left\lceil t\right\rceil -1}{t-1}\right)\right\} \\
 & \leqslant\sup_{t\geqslant s}\left\{ \left(\log\left(\left\lceil t\right\rceil -1+c\right)\right)^{-p}\cdot1\cdot\exp\left(2\right)\right\} \\
 & =e^{2}\cdot\left(\log\left(\left\lceil s\right\rceil -1+c\right)\right)^{-p}\\
 & \leqslant e^{2}\cdot\left(\log\left(s-1+c\right)\right)^{-p}.
\end{align*}
In the case where $\gamma\left(n\right)=\left(n+c\right)^{-p}$ for
some $c>0$, write}

\textcolor{black}{
\begin{align*}
\beta\left(s\right) & \leqslant\sup_{t\geqslant s}\inf_{n\geqslant2}\left\{ \left(n+c\right)^{-p}\cdot\left(\frac{n}{t-1}\right)^{-1}\cdot\exp\left(\frac{n}{t-1}\right)\right\} \\
 & \leqslant\sup_{t\geqslant s}\left\{ \left(\left\lceil t\right\rceil -1+c\right)^{-p}\cdot\left(\frac{\left\lceil t\right\rceil -1}{t-1}\right)^{-1}\cdot\exp\left(\frac{\left\lceil t\right\rceil -1}{t-1}\right)\right\} \\
 & \leqslant\sup_{t\geqslant s}\left\{ \left(\left\lceil t\right\rceil -1+c\right)^{-p}\cdot1\cdot\exp\left(2\right)\right\} \\
 & =e^{2}\cdot\left(\left\lceil s\right\rceil -1+c\right)^{-p}\\
 & \leqslant e^{2}\cdot\left(s-1+c\right)^{-p}.
\end{align*}
Finally, in the case where $\gamma\left(n\right)=\exp\left(-n^{\psi}\right)$
with $\psi\in\left(0,1\right)$, then}

\textcolor{black}{
\[
\beta\left(s\right)\leqslant\sup_{t\geqslant s}\inf_{n\geqslant2}\left\{ \exp\left(-n^{\psi}\right)\cdot\left(\frac{n}{t-1}\right)^{-1}\cdot\exp\left(\frac{n}{t-1}\right)\right\} .
\]
Now, take $n=\left\lceil C\cdot\left(t-1\right)^{\frac{1}{1-\psi}}\right\rceil $
to see that}

\textcolor{black}{
\begin{align*}
\beta\left(s\right) & \leqslant\sup_{t\geqslant s}\left\{ \exp\left(-\left\lceil C\cdot\left(t-1\right)^{\frac{1}{1-\psi}}\right\rceil ^{\psi}\right)\cdot\left(\frac{\left\lceil C\cdot\left(t-1\right)^{\frac{1}{1-\psi}}\right\rceil }{t-1}\right)^{-1}\cdot\exp\left(\frac{\left\lceil C\cdot\left(t-1\right)^{\frac{1}{1-\psi}}\right\rceil }{t-1}\right)\right\} \\
 & \leqslant\sup_{t\geqslant s}\left\{ \exp\left(-C^{\psi}\cdot\left(t-1\right)^{\frac{\psi}{1-\psi}}\right)\cdot\left(C\cdot\left(t-1\right)^{\frac{\psi}{1-\psi}}\right)^{-1}\cdot\exp\left(C\cdot\left(t-1\right)^{\frac{\psi}{1-\psi}}+\frac{1}{t-1}\right)\right\} \\
 & \leqslant\frac{e}{C}\cdot\sup_{t\geqslant s}\left\{ \frac{\exp\left(-\left(C^{\psi}-C\right)\cdot\left(t-1\right)^{\frac{\psi}{1-\psi}}\right)}{\left(t-1\right)^{\frac{\psi}{1-\psi}}}\right\} .
\end{align*}
Let $C\in\left(0,1\right)$ so that $C^{\psi}-C$ is maximised (and
in particular, is positive) and takes the value $C_{\psi}$, so that}

\textcolor{black}{
\begin{align*}
\beta\left(s\right) & \leqslant\frac{e}{C}\cdot\sup_{t\geqslant s}\left\{ \frac{\exp\left(-C_{\psi}\cdot\left(t-1\right)^{\frac{\psi}{1-\psi}}\right)}{\left(t-1\right)^{\frac{\psi}{1-\psi}}}\right\} \\
 & =\frac{e}{C}\cdot\frac{\exp\left(-C_{\psi}\cdot\left(s-1\right)^{\frac{\psi}{1-\psi}}\right)}{\left(s-1\right)^{\frac{\psi}{1-\psi}}}\\
 & \leqslant\frac{e}{C}\cdot\exp\left(-C_{\psi}\cdot\left(s-1\right)^{\frac{\psi}{1-\psi}}\right).
\end{align*}
}

\subsection{Proofs for Section~\ref{subsec:Illustration:-Independent-MH}}
\begin{proof}[Proof of Proposition~\ref{prop:IMH-WPI}.]
\textcolor{black}{We compute directly:
\begin{align*}
\|f\|_{2}^{2} & =\frac{1}{2}\int_{\E\times\E}\dif x\,\dif y\,\pi(x)\pi(y)\left[f(y)-f(x)\right]^{2}\\
 & =\frac{1}{2}\int_{A(s)}\dif x\,\dif y\,\pi(x)\pi(y)\left[f(y)-f(x)\right]^{2}\\
 & \quad+\frac{1}{2}\int_{A(s)^{\complement}}\dif x\,\dif y\,\pi(x)\pi(y)\left[f(y)-f(x)\right]^{2}\\
 & \le\frac{s}{2}\int_{A(s)}\dif x\,\dif y\,\pi(x)\pi(y)\left(w^{-1}(x)\wedge w^{-1}(y)\right)\left[f(y)-f(x)\right]^{2}\\
 & \quad+\frac{1}{2}\int_{A(s)^{\complement}}\dif x\,\dif y\,\pi(x)\pi(y)\left[f(y)-f(x)\right]^{2}\\
 & =\frac{s}{2}\int_{\E\times\E}\dif x\,\dif y\,\pi(x)\pi(y)\left(w^{-1}(x)\wedge w^{-1}(y)\right)\left[f(y)-f(x)\right]^{2}\\
 & \quad+\frac{1}{2}\int_{A(s)^{\complement}}\dif x\,\dif y\,\pi(x)\pi(y)\left[f(y)-f(x)\right]^{2}\left(1-s\left(w^{-1}(x)\wedge w^{-1}(y)\right)\right)\\
 & \le s\mathcal{E}(P,f)+\frac{\pi\otimes\pi(A(s)^{\complement})}{2}\|f\|_{\mathrm{osc}}^{2}.
\end{align*}
}
\end{proof}
\begin{proof}[Proof of Lemma~\ref{lem:indep_ex}]
Direct calculation:
\begin{align*}
\frac{\pi\otimes\pi(A(s)^{\complement})}{2} & =\frac{1}{2}\left[1-\int_{A(s)}\dif x\,\dif y\,a_{1}^{2}\exp\left(-a_{1}(x+y)\right)\right]\\
 & =\frac{1}{2}\left[1-\int_{0}^{\frac{\log s}{a_{2}-a_{1}}}\dif x\int_{0}^{\frac{\log s}{a_{2}-a_{1}}}\dif y\,a_{1}^{2}\exp\left(-a_{1}(x+y)\right)\right]\\
 & =\frac{1}{2}\left[1-\int_{0}^{\frac{\log s}{a_{2}-a_{1}}}\dif x\,a_{1}\exp(-a_{1}x)\left[-\exp(-a_{1}y)\right]_{0}^{\frac{\log s}{a_{2}-a_{1}}}\right]\\
 & =\frac{1}{2}\left[1-\int_{0}^{\frac{\log s}{a_{2}-a_{1}}}\dif x\,a_{1}\exp(-a_{1}x)\left(1-s^{-\frac{a_{1}}{a_{2}-a_{1}}}\right)\right]\\
 & =\frac{1}{2}\left[1-\left(1-s^{-\frac{a_{1}}{a_{2}-a_{1}}}\right)^{2}\right].
\end{align*}
The inequality follows from $1-\left(1-s^{-\frac{a_{1}}{a_{2}-a_{1}}}\right)^{2}=\left(2-s^{-\frac{a_{1}}{a_{2}-a_{1}}}\right)s^{-\frac{a_{1}}{a_{2}-a_{1}}}$.
\end{proof}
\begin{rem}[in relation to Proposition~\ref{prop:Phi-is-2p-norm}]
 \textcolor{black}{In fact, given the invariance of the original
integral under shifts of the form $f\leftarrow f+c$, one can always
refine the above estimate to }

\textcolor{black}{
\[
\int_{A}\mu\left(dx\right)P\left(x,{\rm d}y\right)\left(f\left(x\right)-f\left(y\right)\right)^{2}\leqslant\mu\otimes P\left(A\right)^{1/q}\cdot\tilde{\Phi}\left(f\right),
\]
where
\[
\tilde{\Phi}\left(f\right):=\inf_{m\in\mathbb{R}}\left\{ \Phi\left(f-m\right)\right\} .
\]
One can also verify that $\tilde{\Phi}$ is non-expansive under the
action of $P$ by noting that $P\left(f-m\right)=Pf-m$, and hence
that}

\textcolor{black}{
\begin{align*}
\tilde{\Phi}\left(Pf\right) & =\inf_{m\in\mathbb{R}}\left\{ \Phi\left(Pf-m\right)\right\} \\
 & =\inf_{m\in\mathbb{R}}\left\{ \Phi\left(P\left(f-m\right)\right)\right\} \\
 & \leqslant\inf_{m\in\mathbb{R}}\left\{ \Phi\left(f-m\right)\right\} \\
 & =\tilde{\Phi}\left(f\right).
\end{align*}
}
\end{rem}

\section{Proofs for Section \ref{sec:Application-to-pseudo-marginal}}
\begin{proof}[Proof of Lemma~\ref{lem:averaging}.]
 We centre the $W_{i}$, which are nonnegative, and then use the
binomial theorem: writing $\overline{W}_{i}:=W_{i}-1$ for each $i$,
\begin{align*}
\mathbb{E}\left[\left|\frac{1}{N}\sum_{i=1}^{N}W_{i}\right|^{p}\right] & =\mathbb{E}\left[\left(1+\frac{1}{N}\sum_{i=1}^{N}\overline{W}_{i}\right)^{p}\right]\\
 & =1+\sum_{k=2}^{p}\left(\begin{array}{c}
p\\
k
\end{array}\right)\Ebb\left[\left(\frac{1}{N}\sum_{i=1}^{N}\overline{W}_{i}\right)^{k}\right]\\
 & \le1+\sum_{k=2}^{p}\left(\begin{array}{c}
p\\
k
\end{array}\right)\Ebb\left[\left|\frac{1}{N}\sum_{i=1}^{N}\overline{W}_{i}\right|^{k}\right],
\end{align*}
where we used the fact that each $\Ebb\overline{W}_{i}=0$ to cancel
the $k=1$ summand. We now make use of the Marcinkiewicz--Zygmund
inequality, which tells us that for $k\ge2$, there exist universal
constants $\left\{ B_{k}\right\} $, such that
\begin{align*}
\Ebb\left[\left|\frac{1}{N}\sum_{i=1}^{N}\overline{W}_{i}\right|^{k}\right] & \le B_{k}\Ebb\left[\left(\frac{1}{N^{2}}\sum_{i=1}^{N}\overline{W}_{i}^{2}\right)^{k/2}\right]\\
 & \le B_{k}N^{-k/2}\Ebb\left[\left|\overline{W}_{1}\right|^{k}\right],
\end{align*}
where for the latter inequality we use Minkowski's inequality with
exponent $k/2$. This gives us the bound, for some constants $\left\{ C_{p,k}\right\} $,
\[
\mathbb{E}\left[\left|\frac{1}{N}\sum_{i=1}^{N}W_{i}\right|^{p}\right]\le1+\sum_{k=2}^{p}N^{-k/2}C_{p,k}\Ebb\left[\left|\overline{W}_{1}\right|^{k}\right].
\]
\end{proof}
\begin{proof}[Proof of Proposition \ref{prop:ABC-example}.]
 Since under $Q_{x}$, $\ell_{{\rm ABC}}(x)\sum_{j=1}^{N}W(z_{j})\sim\mathrm{Bin}\big(N,\ell_{{\rm ABC}}(x)\big)$,
we can write
\[
Q_{x}\left[\mathcal{W}_{N}^{p+1}\right]=\frac{\mathbb{E}\left[\mathrm{Bin}(N,\ell_{{\rm ABC}}(x))^{p+1}\right]}{(N\ell_{{\rm ABC}}(x))^{p+1}}.
\]
The (non-centered) moments of a binomial random variable are known
\cite{Knoblauch08} to have the form,
\[
\mathbb{E}\left[\mathrm{Bin}(N,\ell_{{\rm ABC}}(x))^{p+1}\right]=\sum_{k=1}^{p+1}c_{p+1,k}N^{\underline{k}}\ell_{{\rm ABC}}(x)^{k},
\]
for appropriate coefficients $c_{p+1,k}$ and where $N^{\underline{k}}=N(N-1)\cdots(N-k+1)$
is the $k$th falling power of $N$. The result follows immediately
from this identity. 
\end{proof}
\begin{proof}[Proof of Proposition \ref{prop:prod-avg-bound}.]
 We observe that by independence, we may write for a given $x\in\mathsf{X}$,
\[
Q_{x}(W^{p})=Q_{x}\left(\prod_{t=1}^{T}W_{t}^{p}\right)=\prod_{t=1}^{T}Q_{x,t}\left(W_{t}^{p}\right),
\]
for some distributions $Q_{x,1},\ldots,Q_{x,T}$. By Lemma~\ref{lem:averaging},
we have (with dependence on $t$ and $x$ suppressed from $C_{p,k}$)
\begin{align*}
Q_{x,t}\left(W_{t}^{p}\right) & \leq1+\sum_{k=2}^{p}N^{-k/2}C_{p,k}Q_{x,t}\left[|W_{t}-1|^{k}\right]\\
 & \leq1+\max_{k\in\{2,\ldots,p\}}\left\{ C_{p,k}Q_{x,t}\left[|W_{t}-1|^{k}\right]\right\} \sum_{k=2}^{p}N^{-k/2}\\
 & =:1+M_{t,p}\sum_{k=2}^{p}N^{-k/2},
\end{align*}
for some constants $\{C_{p,k}\}$ and $M_{t,p}$. Since $\sum_{k=2}^{\infty}N^{-k/2}=1/(N-\sqrt{N})$,
we deduce that there exists functions $M_{t,p}\leq M_{p}$ such that
for any $x\in\mathsf{X}$,
\[
Q_{x}(W^{p})=Q_{x}\left(\prod_{t=1}^{T}W_{t}^{p}\right)\leq\prod_{t=1}^{T}\left(1+\frac{M_{t,p}(x)}{N-\sqrt{N}}\right)\leq\left(1+\frac{M_{p}(x)}{N-\sqrt{N}}\right)^{T}.
\]
For the given choice of $N$,
\[
Q_{x}(W^{p})\leq\left(1+\frac{M_{p}(x)}{\alpha T}\right)^{T}\leq\exp\left(\frac{M_{p}(x)}{\alpha}\right).
\]
Hence, we obtain,
\begin{align*}
\int\pi({\rm d}x)\tilde{\pi}_{x}(W\geq s) & \leq\int\pi({\rm d}x)\frac{\tilde{\pi}_{x}(W^{p-1})}{s^{p-1}}\\
 & =\int\pi({\rm d}x)\frac{Q_{x}(W^{p})}{s^{p-1}}\\
 & \leq s^{-p+1}\int\pi({\rm d}x)\exp\left(\frac{M_{p}(x)}{\alpha}\right),
\end{align*}
as required.
\end{proof}
\begin{proof}[Proof of Lemma \ref{lem:beta-s-log-normal}.]
We calculate directly: recalling that $\tilde{\pi}_{x}(\dif w)=wQ_{x}(\dif w)$,
\begin{align*}
\tilde{\pi}_{x}(W & \ge s)=\int_{s}^{\infty}\frac{1}{\sigma\sqrt{2\pi}}\exp\left(-\frac{\left(\log w+\sigma^{2}/2\right)^{2}}{2\sigma^{2}}\right)\,\dif w\\
 & =\frac{1}{\sigma\sqrt{2\pi}}\int_{\log s}^{\infty}\exp\left(-\frac{\left(z-\sigma^{2}/2\right)^{2}}{2\sigma^{2}}\right)\,\dif z\\
 & =\mathbb{P}\left(Z\ge\frac{\log s-\sigma^{2}/2}{\sigma}\right),
\end{align*}
where we used the substitution $z=\log w$, and in the final expression
$Z\sim\mathcal{N}(0,1)$. The result then follows from standard sub-Gaussian
tail bounds, e.g. \cite[Chapter 2, Prop. 2.5]{Wainwright19}.
\end{proof}
\begin{proof}[Proof of Lemma \ref{lem:log_normal_lambert}.]
 This follows by first bounding $F(w)$ using Lemma\ \ref{lem:boundK*}:
\begin{align*}
F(w) & =\int_{w}^{1}\frac{\dif v}{K^{*}(v)}\\
 & \le2\exp(\sigma^{2}/2)\int_{w}^{1}\frac{1}{v}\exp\left(\sigma\sqrt{-2\log\frac{v}{2}}\right)\,\dif v\\
 & =2\exp(\sigma^{2}/2)\int_{-\log\frac{1}{2}}^{-\log\frac{w}{2}}\exp\left(\sigma\sqrt{2z}\right)\,\dif z\\
 & =2\exp(\sigma^{2}/2)\left[\sigma^{-2}\exp\left(\sigma\sqrt{2z}\right)\left(\sigma\sqrt{2z}-1\right)\right]_{-\log\frac{1}{2}}^{-\log\frac{w}{2}}\\
 & \le2\exp(\sigma^{2}/2)\sigma^{-2}\exp\left(\sigma\sqrt{-2\log\frac{w}{2}}\right)\sigma\sqrt{-2\log\frac{w}{2}}.
\end{align*}
where we made use of the substitution $z=-\log\frac{v}{2}$. The result
then follows from inverting this relationship, using the fact that
$\mathsf{W}$ is the inverse of the map $x\mapsto x\exp(x)$.
\end{proof}
\begin{proof}[Proof of Proposition \ref{prop:mixing-time-lognormal}.]
 For notational simplicity we may drop the argument of some of the
functions involved. Let $C=\mathsf{B}C_{\mathrm{P}}/2\sigma_{0}^{2}$
then the constraint gives
\[
-\frac{1}{\sigma^{2}}\mathsf{W}^{2}\left(\frac{C\sigma^{4}}{\exp(\sigma^{2}/2)}\right)=-H(\epsilon)=-2\log\left(\frac{2}{C_{{\rm P}}\epsilon^{2}}\right).
\]
We may express $C$ in terms of $\sigma$,
\[
C(\sigma)=\sqrt{H}\frac{\exp\left(\sigma^{2}/2+\sigma\sqrt{H}\right)}{\sigma^{3}},
\]
which we optimise w.r.t. $\sigma$, to obtain the minimiser $\bar{\sigma}_{\star}=(-\sqrt{H}+\sqrt{H+12})/2.$
Since 
\[
\frac{6}{H(1+12/H)^{1/2}}\leq\frac{1}{2}\int_{0}^{12/H}(1+u)^{-1/2}{\rm d}u\leq\frac{6}{H},
\]
for $H\geq1$,
\[
\frac{1}{4}\frac{3}{\sqrt{H}}\leq\frac{3}{\sqrt{H}(1+12/H)^{1/2}}\leq\bar{\sigma}_{\star}\leq\frac{3}{\sqrt{H}},
\]
while $\lim_{H\to\infty}\sqrt{H}\bar{\sigma}_{\star}=3$. We may bound
the budget by taking $\bar{\sigma}(\epsilon)=3/\sqrt{H(\epsilon)}$
so $\bar{N}(\epsilon)=\sigma_{0}^{2}H(\epsilon)/9$,
\[
\mathsf{B}(\epsilon)=\frac{2\sigma_{0}^{2}}{C_{{\rm P}}}C(\sigma(\epsilon))=\frac{2\sigma_{0}^{2}H(\epsilon)^{2}\exp\left(\frac{9}{2H}+3\right)}{27C_{{\rm P}}}\leq\frac{2\sigma_{0}^{2}H(\epsilon)^{2}\exp\left(15/2\right)}{27C_{{\rm P}}},
\]
and we have 
\[
\bar{n}(\epsilon)\leq\frac{2H(\epsilon)\exp\left(15/2\right)}{3C_{{\rm P}}}.
\]
Taking $H(\epsilon)=2\log\left(\frac{2}{C_{{\rm P}}\epsilon^{2}}\right)$
then gives the result.
\end{proof}
\begin{proof}[Proof of Lemma \ref{lem:bound-asymp-var}.]
 We use the following expression for the asymptotic variance,
\begin{align*}
v(f,P) & =2\langle(\Id-P)^{-1}f,f\rangle-\|f\|_{2}^{2}\\
 & =2\sum_{n=0}^{\infty}\langle f,P^{n}f\rangle-\|f\|_{2}^{2}\\
 & =2\sum_{l=0}^{\infty}\left\langle P^{l}f,({\rm Id}+P)P^{l}f\right\rangle -\|f\|_{2}^{2}\\
 & \le4\sum_{l=0}^{\infty}\|P^{l}f\|_{2}^{2}-\|f\|_{2}^{2}\\
 & \le4\sum_{l=0}^{\infty}\Phi(f)F^{-1}(l)-\|f\|_{2}^{2}.
\end{align*}
\end{proof}

\section*{Acknowledgments}

We would like to thank Gareth Roberts \textcolor{black}{and Chris
Sherlock} for useful comments. Research of CA, AL and AQW supported
by EPSRC grant `CoSInES (COmputational Statistical INference for Engineering
and Security)' (EP/R034710/1), and research of CA and SP supported
by EPSRC grant Bayes4Health, `New Approaches to Bayesian Data Science:
Tackling Challenges from the Health Sciences' (EP/R018561/1).

\end{document}